\documentclass[11pt]{article}
\usepackage[T1]{fontenc}
\usepackage{csquotes}
\usepackage{tgheros}

\usepackage{comment}

\usepackage{amsmath,amsfonts,amsthm,mathtools,color}
\usepackage{mathtools,color}
\usepackage{fullpage}
\usepackage{thmtools,thm-restate}
\usepackage[algo2e,linesnumbered,noend,ruled,noline]{algorithm2e} 
\usepackage{setspace}
\usepackage{tikz}

\SetAlFnt{\small}
\SetAlCapFnt{\small}
\SetAlCapNameFnt{\small}
\SetAlCapHSkip{0pt}
\IncMargin{-\parindent}

\usepackage{graphicx}
\usepackage{amsthm}
\usepackage{nicefrac}
\usepackage{hyperref}
\usepackage{cleveref}
\crefname{theorem}{Theorem}{Theorems}
\crefname{lemma}{Lemma}{Lemmata}
\crefname{claim}{Claim}{Claims}
\crefname{definition}{Definition}{Definitions}
\crefname{appendix}{Appendix}{Appendices}
\crefname{example}{Example}{Examples}
\crefname{equation}{Equation}{Equations}
\crefname{section}{Section}{Sections}
\crefname{subsection}{Subsection}{Subsections}

\usepackage[numbers]{natbib}
\setcitestyle{acmnumeric}
\DontPrintSemicolon

\SetCommentSty{mycommfont}
\SetArgSty{textnormal}
\usepackage{soul}
\usepackage[title]{appendix}
\usepackage{apxproof}
\usepackage{cancel}
\usepackage{authblk}
\usepackage[normalem]{ulem}
\usepackage{xcolor}

\newtheorem{theorem}{Theorem}
\newtheorem{example}{Example}
\newtheorem{lemma}[theorem]{Lemma}
\newtheorem*{lemma*}{Lemma}

\newtheorem*{theorem*}{Theorem}

\newtheorem{claim}[theorem]{Claim}
\newtheorem{definition}[theorem]{Definition}
\usepackage[inline]{enumitem}

\usepackage{authblk}
\newcommand{\E}{\mbox{\bf E}}

\DeclareFontShape{T1}{cmr}{m}{scit}{<->ssub*cmr/m/sc}{}

\title{On the Power of Randomization for Obviously Strategy-Proof Mechanisms}

\author[1]{Shiri Ron}
\author[2]{Daniel Schoepflin}

\affil[1]{Weizmann Institute of Science}
\affil[2]{Rutgers University}

\begin{document}
\titlepage
\maketitle
\begin{abstract}
We investigate the problem of designing \emph{randomized} obviously strategy-proof (OSP) mechanisms in several  canonical  auction settings. 
Obvious strategy-proofness, introduced by Li \cite{li}, strengthens the well-known concept of dominant-strategy incentive compatibility (DSIC). Loosely speaking, it ensures that even agents who struggle with contingent reasoning can identify that their dominant strategy is optimal.

Thus, one would hope to design OSP mechanisms with good approximation guarantees.
Unfortunately, 
deterministic OSP mechanisms fail to achieve an approximation better than $\min\{m,n\}$ where $m$ is the number of items and $n$ is the number of bidders, even 
for the simple settings of additive and unit-demand bidders \cite{Ron24}. 
We circumvent these impossibilities
by showing that randomized
mechanisms that are obviously strategy-proof in the universal sense obtain a \emph{constant} factor approximation for these classes. 
{We show that this phenomenon occurs also for the setting of a multi-unit auction with single-minded bidders.}
Thus, our results provide a more positive outlook on the design of OSP mechanisms and exhibit a stark separation between the power of randomized and deterministic OSP mechanisms.


To complement the picture, we provide impossibilities for randomized OSP mechanisms in each setting. While the deterministic VCG mechanism is well known to output an optimal allocation in dominant strategies, we show that even randomized OSP mechanisms cannot obtain more than $87.5\%$ of the optimal welfare. This further demonstrates that OSP mechanisms are significantly weaker than dominant-strategy mechanisms.


\end{abstract}
\clearpage
\newcommand{\bidders}{\mathcal{N}}
\newcommand{\items}{M}
\newcommand{\opt}{\text{OPT}}
\newcommand{\allocs}{\mathcal{T}}
\newcommand{\detmech}{A}
\newcommand{\mechanism}{\mathcal{M}}
\newcommand{\random}{R}
    
\section{Introduction}
Economics is the science of how to allocate scarce resources to several competing parties. In particular, auctions serve as a useful playground to understand who should get what and for what price.  
We assume a {good-willed} central planner who aims to allocate the resources in a way that maximizes the social welfare of all parties involved.  To achieve that, she has to overcome the following obstacle: the information of bidders is private and they are 
 interested in maximizing their own utility. Therefore, she must carefully design the elicitation mechanism to align the incentives of the agents with her own objective.



For a specific example, consider the  setting where an auctioneer 
wants to allocate a single item among a set $\bidders$ of  bidders.  
Each bidder $i$ has a value $v_i$ for receiving the good.  To maximize social welfare, the auctioneer should give the item to the bidder of highest value, i.e., $i^* = \text{argmax}_{i \in \bidders}{v_{i}}$.  
For that, the auctioneer must design a mechanism which collects
information about the value of the bidders
and decides which bidder wins
the good and the payment $p_i \geq 0$ of each bidder.

The classic solution proposed for this 
setting is the \emph{sealed-bid} second-price auction, wherein bidders report their values directly to the auctioneer and the highest-valued bidder is awarded the good at the second highest price \cite{Vic61}.  It is well-known that this auction is \emph{dominant-strategy incentive compatible} (i.e., \emph{strategy-proof}), meaning that each bidder maximizes her utility by truthfully reporting her private value regardless of the reports of the other bidders.  Theory suggests, therefore, that bidders should never misreport their value in this auction.  However, in practice, 
``real-world'' bidders report bids not equal to their true value \cite{KHL1987}.  Thus, there appears to be a mismatch between the prediction of the theory of strategy-proof mechanisms and the observed outcomes. 

An alternative to the sealed-bid second-price auction for the single-item setting is the ascending price (Japanese) auction.  In this auction, a price clock gradually increases over time and bidders drop out whenever the asking price becomes too high.  The ascending price auction implements the exact same outcome as the sealed-bid auction: it awards the item to the highest-valued bidder at the second-highest value. However, bidders appear empirically  more likely to follow their optimal truthful strategy when facing an ascending-price auction compared to the sealed-bid format.

To address this discrepancy, Li introduced the notion of \emph{obvious strategy-proofness} (OSP), a strengthening of strategy-proofness \cite{li}. Loosely speaking, an OSP mechanism ensures that even agents unable to perform contingent reasoning can recognize truth-telling as the optimal strategy. OSP provides a theoretical explanation for the prevalence of ascending auctions over sealed-bid implementations, by claiming that their popularity stems from the fact that they are simpler for bidders to understand than sealed-bid implementations.


Since its introduction by \cite{li}, obvious strategy-proofness has emerged as a ``gold standard'' for strategic simplicity in mechanism design and the notion has attracted a great deal of attention.  For instance, various refinements and relaxations of OSP have been proposed (e.g., \cite{pycia2023theory,nagel2023measure,FV23}) and the design of OSP mechanisms has been examined various settings, including, e.g., one-sided matching \cite{troyan2019obviously} and two-sided  matching  markets \cite{AG18,Thomas21}, scheduling \cite{ferraioli2019obviously}, voting problems \cite{BG17,arribillaga2020obvious} and allocation problems \cite{BG17}.  Another line of work has aimed to find \emph{characterizations} of OSP mechanisms in various domains.  For instance, \cite{li} showed that for all single-parameter binary allocation settings the class of OSP mechanisms coincides with the class of ``personalized clock auctions'' (essentially, a natural generalization of the Japanese auction for a single item), and further work has  established characterizations for all linear single-parameter domains \cite{carmine1,carmine2} and beyond \cite{pycia2023obviously,Mackenzie20}.  Using these characterizations, both lower bounds (impossibility results) and upper bounds (mechanisms with proven guarantees) have been proposed for various single-parameter auction settings, including binary allocation with general feasibility constraints \cite{DGR14,CGS22,FGGS22}  and procurement  settings \cite{BGGST22}.  

Beyond single-parameter settings, the picture regarding the performance of OSP mechanisms in auctions becomes somewhat pessimistic.
\cite{BG17} initially showed that, for additive bidders, no obviously strategy-proof mechanism optimizes the welfare,
and recent work of \cite{Ron24} has, essentially, ``closed the book'' on deterministic OSP mechanisms for multi-parameter combinatorial auction settings.  Even if bidders have additive or unit-demand valuation functions (which are commonly thought to be  ``easy''), the trivial OSP mechanism which 
runs an ascending-price auction for the grand bundle achieves the \emph{best possible} approximation guarantee of $\min\{m,n\}$, where $m$ is the number of items and $n$ is the number of bidders. 
To  circumvent these strong impossibilities, we turn our attention to \emph{randomized} OSP mechanisms.

\subsection{Our Results}
Our main results are upper bounds and lower bounds for randomized OSP mechanisms.
We focus our attention on ``universally'' OSP mechanisms, i.e., mechanisms which are a distribution over deterministic OSP mechanisms.
We analyze all the settings considered by 
 \cite{Ron24} and  show that:
\begin{enumerate}
    \item For additive bidders in a combinatorial auction, there is a mechanism that obtains a $4$ approximation and
    no mechanism has approximation better than $\frac{8}{7}\approx 1.14$  (\cref{thm-ub-add} and \cref{thm-lb-add}).
\item For unit-demand bidders in a combinatorial auction, 
there is a mechanism
    that obtains an $e\approx 2.72$ approximation and
no mechanism has approximation better than
than $\frac{8}{7}\approx 1.14$ (\cref{thm:ud-upper} and \cref{thm-lb-ud}). 
    \item For single-minded bidders in a multi-unit auction with unknown demands, 
    there is a mechanism that obtains a $400$ approximation and no mechanism has approximation better than $1.2$ (\cref{thm-ub-mua-sm} and \cref{thm-mua-sm-lb}). 
\end{enumerate}
All the impossibilities 
are for mechanisms that satisfy individual rationality and no negative transfers. Likewise,
our proposed mechanisms
conform to these conditions.

Observe that our upper bounds demonstrate the power of randomization for obviously strategy-proof mechanism design:
whilst deterministic OSP mechanisms 
can only obtain an approximation of 
$\{m,n\}$ to the optimal welfare \cite{Ron24}, 
all these classes admit a randomized OSP mechanisms with a  constant factor approximation. 
In addition, we observe that the randomized $poly(m)$-communication mechanisms that are
dominant-strategy incentive compatible and obtain the state of the art approximation guarantees for \textquote{richer} classes of valuations in combinatorial auctions are in fact  obviously strategy-proof (see \cref{cl:subadditive,cl:general}).\footnote{We also provide a $400$ approximation to the optimal welfare for
multi-unit auctions with bidders whose valuations satisfy decreasing marginal utilities (\cref{thm:decreasing-marginals}). This is the only multi-parameter domain for which the power of deterministic mechanisms is not known. In \cref{subsub::non-mono}, 
we describe a non-monotonicity effect that illustrates a barrier towards proving impossibilities for this class.}


{Our upper bounds are motivated by the following observation:
the constructions of \cite{Ron24}
show 
that
every  mechanism that provides a non-trivial approximation to the welfare satisfies that the first bidder that \textquote{speaks} in the mechanism does not have an obviously dominant strategy. The underlying cause of this  phenomenon is that when querying a bidder for the first time, the mechanism fails because it has no information regarding the valuations of the other bidders. 
Thus, to overcome this impossibility, our proposed mechanisms are based on the classic secretary approach of sampling a sufficient fraction of the bidders and aggregating their information to determine
a price per item. Owing to the use of randomization, this can be done 
in an obviously dominant manner while maintaining a high fraction of the welfare in expectation.}  



Our lower bounds for combinatorial auctions with unit-demand and additive bidders further emphasize 
the restrictiveness of obvious strategy-proofness compared to implementation in dominant strategies. 
Not only getting an approximation better than $\min\{m,n\}$ is impossible deterministically, but even if we allow randomization we cannot get more than $87.5\%$ of the optimal welfare. In contrast, these settings  have  dominant-strategy mechanisms that extract the optimal welfare and are also efficient both from a computation and a communication perspective.
One disadvantage of our main results is that the lower bounds and upper bounds that we provide are quite far apart. As a step to bridge this gap, we show in 
\cref{sec-22} 
that for two bidders and two items, all the aforementioned classes
admit mechanisms that give a $\frac{4}{3}$ approximation.  

\subsection{Why Randomization?}
On first impression, one might argue that randomization adds impractical complexity to a mechanism, and does not align with the simplicity we aim to achieve when designing OSP mechanisms.  Indeed,  ``real-world'' agents could possibly be confused by randomization. 
Moreover, it can be difficult to verify that an outcome is the result of some pre-specified random process.  We emphasize, however, that the randomization at use in our work is rather ``straightforward'' in the sense that bidders do not reason about 
expected outcomes since, on any fixed result of the random process, they face a mechanism where they have an obviously dominant strategy.

Moreover, randomized mechanisms are prevalent in practice, e.g., in drafts in sports for new team members, housing programs, and ``greencard'' allocation, and well-studied elsewhere in theory, in particular in fair division (see \cite{budish2013designing} among many others). 
In fact, designing fair mechanisms without randomization is challenging, as the order of player selection often plays a crucial role—bidders who are selected earlier tend to have an advantage over those selected later. Randomizing the order of choice is a natural solution, and there appears to be no alternative to this approach.

The power of randomization for welfare maximization is,  despite a great deal of work, still not fully understood. If we put aside complexity considerations, then the deterministic VCG mechanism is optimal, so randomized dominant-strategy mechanisms do not \textquote{beat} their randomized counterparts.  

In contrast, 
 if we require communication-efficient mechanisms, i.e. mechanisms with number of bits that is polynomial in $m$ and $n$ in the worst case, then 
 randomized dominant-strategy mechanisms achieve an approximation of $\mathcal O(\sqrt m)$ when the bidders have arbitrary monotone valuations, whilst every deterministic mechanism that achieves approximation better than $\mathcal O(m^{1-\epsilon})$ has exponential communication \cite{dobzinski2012truthful,DRV22}. If we settle for mechanisms that only satisfy the weaker notion of ex-post incentive compatibility\footnote{A mechanism is ex-post incentive compatible if it has strategies that form a Nash equilibrium, in contrast to dominant-strategy mechanisms where the bidders have dominant strategies.}, then
the best deterministic mechanisms lag behind their randomized counterparts \cite{QW24,AS19,assadi2021improved}, but no such separation is known. Our work, in contrast, demonstrates such a separation between deterministic and randomized OSP mechanisms.




\subsection{Open Questions}

{We note that there is a gap between our upper and lower bounds and we} leave open 
the question of understanding the exact approximation ratio of these classes of mechanisms and of randomized OSP mechanisms in general.  A particular enticing question is whether there exists an $\mathcal O(1)$-approximate randomized OSP mechanism for more general complement-free valuations (e.g., submodular or subadditive valuations) or whether one can demonstrate an $\omega(1)$-lower bound on the performance of any OSP mechanism in these settings.  

Toward the latter, one would need to circumvent a crucial limitation of our lower bound approach. The reason for it is that we  prove our lower bounds
by presenting a distribution over valuation profiles, and showing that a mechanism has to fail on at least one valuation profile.  To prove super-constant lower bounds, it is essential to find distributions whose support has super-constant number of instances and show that every deterministic mechanism 
errs on a significant fraction of them.

There are also interesting  questions outside the realm of the auction settings.  For example, exploring other settings to determine whether there exists a separation between the performance
dominant-strategy and obviously strategy-proof mechanisms is an intriguing direction.



\section{Preliminaries and Useful Observations}\label{sec-prelims}
Let $\items$ be the set of $m$ items and $\bidders$ the set of $n$ bidders. Each bidder $i$ has a private valuation function $v_i: 2^{\items} \to \mathbb{R}^{+}$, which specifies her value for every subset of items. This function belongs to a domains of valuations  $V_i$ that is known to the auctioneer.  
The auctioneer's goal is to maximize social welfare, i.e., to output a partition $(S_1, \dots, S_n)$ of the items $\items$ among the bidders $\bidders$ that maximizes $\sum_i v_i(S_i)$.


To achieve this, the auctioneer designs mechanisms, with this work focusing specifically on randomized ones.
A randomized protocol $\mathcal M$ randomly chooses, in advance, one of several deterministic protocols to follow. We denote the deterministic mechanisms in the support of the randomized mechanism with $\mathcal A$.
Throughout the paper, we use the terms protocol and mechanism interchangeably.
We defer all proofs to \cref{app:missing-proofs-prems}. 

\paragraph{Components of Deterministic Protocols} Following the random selection, all bidders face a deterministic protocol.
We represent protocols as trees where each internal node corresponds to a bidder called to "speak" or communicate a message.
Each such node has a set of possible messages, and the next node in the tree is determined by the message sent by the bidder. 
A leaf in the protocol specifies an allocation of items to bidders and a payment for each bidder.

Fixing a randomly chosen deterministic protocol $A\in \mathcal A$, let $\mathcal{N}_i$ denote the set of all nodes in which a particular bidder $i$ is called to speak.  Then, the behavior $B_i$ of player $i$ assigns a message to each node in $\mathcal{N}_i$. We denote with $\mathcal{B}_i$ be the set of all possible behaviors.  Observe that a behavior profile $B = (B_1,\dots, B_n) \in \mathcal B_1\times\cdots\times \mathcal B_n$ thus defines a root to leaf path in $A$.  We let $Path(B)$ denote all the nodes along the path defined by $B$ and $Leaf(B)$ denote the leaf
that $Path(B)$ ends with.  For every behavior profile $B$ and for every player $i$, we 
denote with $f_i(B)$ and with $p_i(B)$ respectively the allocation and the payment of player $i$ that are specified in $Leaf(B)$. For two behavior profiles $B$ and $B'$, we denote all of the nodes appearing in both $Path(B)$ and $Path(B')$ as $Path(B)\cap Path(B')$. Two behavior profiles $(B_1,\ldots,B_n)$ and $(B_1',\ldots,B_n')$ \emph{diverge at vertex $u$} if $u$ is the last shared vertex in $Path(B_1,\ldots,B_n)$ and $Path(B_1',\ldots,B_n')$.

Finally, the \emph{strategy} $\mathcal{S}_i$ of player $i$ is a function specifying a behavior of player $i$ for each possible valuation in $V_i$ and every possible deterministic protocol $A$ in the support $\mathcal A$. Formally, $\mathcal S_i:\mathcal A \times V_i\to\mathcal B_i$. We often abuse notation by referring to the strategy of player $i$ in the deterministic protocol $A$, that is, the partial function $\mathcal{S}_i(A, \cdot)$ simply as a strategy and denoting it by $\mathcal{S}_i$.

A deterministic mechanism $A$ together with strategies $(\mathcal S_1(A,\cdot ),\ldots,\mathcal S_n(A,\cdot))$ realize allocation rule $f:V_1\times \cdots \times V_n \to \allocs$ and payment schemes $P_1,\ldots,P_n:V_1\times \cdots \times V_n \to \mathbb R^n$ if for every $(v_1,\ldots,v_n)\in V_1\times\cdots\times V_n$, it holds that $Leaf(\mathcal S_1(A,v_1),\ldots,\mathcal S_n(A,v_n))$ is labeled with the allocation $f(v_1,\ldots,v_n)$ and with the payment $P_i(v_1,\ldots,v_n)$ for every player $i$. 

\paragraph{Properties of Randomized Mechanisms}
To analyze the performance of our  mechanisms, we compare against the optimal social welfare.  Let $\mathbf{T} = (T_1, \dots ,T_n)$ denote a feasible allocation of the items (i.e., each item is allocated to at most one bidder) and $\allocs$ denote the set of all feasible allocations. 
 Then, we let $\opt(I) = \max_{\mathbf{T} \in \allocs}\sum_{i \in \bidders}{v_i(T_i)}$ denote the optimal social welfare achievable on a given instance $I=(v_1,\ldots,v_n)$ and $\mathbb{E}[W(\mechanism(I))]$ denote the \emph{expected} social welfare achieved by mechanism $\mechanism$ on instance $I$ (where the expectation is taken over the random choice of which deterministic protocol is to be run by the mechanism).  We then say that mechanism $\mechanism$ obtains an $\alpha$-approximation to the optimal social welfare on a class of instances $\mathcal{I}$ if \[\sup_{I \in \mathcal{I}} \frac{\opt(I)}{\mathbb{E}[W(\mechanism(I))]} \leq \alpha.\]

In addition to the objective of welfare maximization, our goal is to design randomized mechanisms  satisfying three key desiderata: (i) \emph{ex-post no negative transfers}; (ii)  \emph{ex-post individual rationality}; and (iii) \emph{universal obvious strategy-proofness}, i.e., ex-post obvious strategy-proofness.\footnote{
A randomized mechanism satisfies a given property \emph{ex-post} if that property holds for every deterministic mechanism that has a non-zero probability of being selected.} For simplicity, in the rest of the paper, we omit the prefix ``ex-post'' when referring to these properties. 

A mechanism $\mechanism$ with support $\mathcal A$ satisfies \emph{no negative transfers} if for every leaf in every deterministic protocol $A\in \mathcal A$, the payment of every player $i$ is at least zero.
A mechanism $\mathcal M$ with strategy profile $(\mathcal S_1,\ldots, \mathcal S_n)$ and support $\mathcal A$ satisfies 
 \emph{individual rationality} if,
for every deterministic protocol $A\in \mathcal A$,
the allocation rule $f:V_1\times \cdots \times V_n\to \allocs$ and payment schemes $P_1,\ldots,P_n$  that it realizes satisfy that 
for every  $(v_1,\ldots,v_n)\in V_1\times \cdots \times V_n$ and every player $i$:
$
v_i(f(v_1,\ldots,v_n))-P_i(v_1,\ldots,v_n)\ge 0
$.
Namely, a mechanism is ex-post individually rational if each player obtains non-negative utility for participating in the mechanism (hence, there is no incentive to avoid participation).  

We will now define obvious strategy-proofness in the universal sense. A mechanism $\mathcal M$ with support $\mathcal A$ and the strategies $(\mathcal S_1,\ldots,\mathcal S_n)$ is universally obviously strategy-proof if for every deterministic protocol $A\in \mathcal A$,
the strategies $(\mathcal S_1(A,\cdot),\ldots,\mathcal S_n(A,\cdot))$ are obviously dominant. 
It remains to define what it means for a strategy to be obviously dominant. Loosely speaking, a strategy 
$\mathcal S_i(A,\cdot)$
of bidder $i$ in a deterministic mechanism $A$ is \emph{obviously dominant} if,
each time player $i$ is called to speak, the worst-case outcome from sending the message defined by 
$\mathcal S_i(A,\cdot)$ 
is weakly better than the best-case outcome from any other strategy.
Thus, it is conceptually ``easy'' for player $i$ to find $\mathcal{S}_i$, understand its dominance and follow it. Despite their intuitive appeal, the definition of obviously
dominant strategies is quite subtle. Thus, 
below, we state only the basic properties of obviously strategy-proof mechanisms that we use to prove our results, and 
we defer the  definition to 
\cref{app-formalities}.


\paragraph{Generalized Ascending Auctions}
To prove some of our positive results, we  employ a specific form of auction, which we name generalized ascending auctions. In particular, some of our randomized mechanisms will be a randomization of such auctions. 

A \emph{generalized ascending auction} defines two possible allocations for each bidder $i$: the \emph{base} bundle $X_i^{B}$ and the \emph{potential} bundle $X_i^{P}$, where $X_i^{B} \subseteq X_i^{P}$. Each bidder $i$ initially ``holds'' the base bundle $X_i^{B}$ and is placed in the ``active'' set $\bidders^{A}$.  
Each $i \in \bidders^{A}$ faces a monotonically increasing price trajectory for receiving $X_i^{P}$ instead of $X_i^{B}$. Bidders drop out when the price for $X_i^{P}$ becomes too high, at which point they are awarded $X_i^{B}$ at a price of $0$ and removed from $\bidders^{A}$ (i.e., they become inactive). The auction terminates when it is feasible to allocate $X_i^{P}$ to all remaining active bidders $i \in \bidders^{A}$ and $X_i^{B}$ to all inactive bidders $i \notin \bidders^{A}$.  
For illustration, an auction where bidder $1$ always receives a fixed item $a$, while the remaining items $M\setminus \{a\}$ are allocated via an ascending auction among the other bidders, is a generalized ascending auction.

\begin{lemma}\label{lemma-partial}
    Every generalized ascending auction is obviously strategy-proof. 
\end{lemma}

\subsection{Tools For Establishing Lower Bounds}
To prove impossibility results for randomized mechanisms, we employ an adaptation of Yao's lemma \cite{Yao83}, which is formalized in \cref{lem:yaos}. This approach allows us to focus on the performance of deterministic mechanisms with obviously dominant strategies when evaluated over a distribution of valuation profiles. We then restate a property of deterministic mechanisms from \cite{Ron24}, which we will use extensively in our impossibility proofs (\cref{lemma-bad-leaf-good-leaf}).


\paragraph{From Randomized Mechanisms to Deterministic Mechanisms}


To state and prove \cref{lem:yaos}, we need the following notations. 
Given a deterministic obviously strategy-proof mechanism $A$ that has obviously dominant strategies $(\mathcal S_1,\allowbreak \ldots, \mathcal S_n)$ and a valuation profile $(v_1,\ldots,v_n)$, we denote with $A(v_1,\ldots,v_n)$ the welfare of the allocation that $A$ outputs given $(\mathcal S_1(v_1),\allowbreak\ldots,\mathcal S_n(v_n))$.
Given a randomized mechanism  $\mathcal M$ 
which is a distribution over such deterministic mechanisms, let $\mathcal M(v_1,\ldots,v_n)$ be the expected welfare of the mechanism given the valuation profile $(v_1,\ldots,v_n)$. 
We denote with $OPT(v_1,\ldots,v_n)$ the optimal welfare.  
\begin{lemma}\label{lem:yaos} \cite{Yao83}
Fix a set of $n$ bidders with domains of valuations $V=V_1\times \cdots \times V_n$.
    Let $\mathcal D$ be a distribution over a set of valuation profiles taken from $V$ and fix an accuracy parameter $\alpha$.
    If for every deterministic mechanism $A$ that is obviously strategy-proof and satisfies no negative transfers and individual rationality, it holds that:
    $$
    \E_{(v_1,\ldots,v_n)\sim \mathcal D}\Big[\dfrac{A(v_1,\ldots,v_n)}{OPT(v_1,\ldots,v_n)}\Big] \le  \frac{1}{\alpha}
    $$
Then, every randomized mechanism $\mathcal M$
that is obviously strategy-proof and satisfies individual rationality and no negative transfers
satisfies that its approximation ratio in the worst case does not exceed $\alpha$. 
\end{lemma}
Note the following subtlety in the statement of the lemma: we consider 
deterministic mechanisms that are obviously strategy-proof, individually rational and 
satisfy no negative transfers with respect to all the valuations in $V$, not only the valuations in the support of $\mathcal D$. Accordingly, our proof implies an impossibility for randomized mechanisms which are a probability distribution over deterministic mechanisms that satisfy all the above properties with respect to all the valuations in $V$. 
A familiar reader may anticipate fully the proof of \cref{lem:yaos}, which we write for completeness in \cref{app:missing-proofs-prems}. 

\paragraph{Proving Lower Bounds For Deterministic Mechanisms} 
Having utilized \cref{lem:yaos} to transition from analyzing randomized mechanisms to analyzing deterministic ones, our next step is to establish lower bounds for the latter. To that end, we invoke \cref{lemma-bad-leaf-good-leaf}, a structural property of deterministic obviously strategy-proof mechanisms originally presented in \cite{Ron24}. While the lemma’s statement may appear technical, it is a natural and intuitive property of obviously strategy-proof mechanisms that stems directly from their definition (see Figure~2 in \cite{Ron24} for an explanatory illustration).  


\begin{lemma}
\label{lemma-bad-leaf-good-leaf}
Fix a deterministic obviously strategy-proof mechanism $A$ with strategies $(\mathcal S_1,\ldots, \mathcal S_n)$
that realize an allocation rule and payment schemes $(f,P_1,\ldots,P_n):V_1\times\cdots \times V_n\to \allocs \times \mathbb R^n$.
Fix a player $i$, a vertex $u\in \mathcal N_i$
and
two valuation profiles $(v_i,v_{-i}),(v_i',v_{-i}')$ such that the following conditions hold simultaneously:
\begin{enumerate}
    \item $u\in Path(\mathcal S_i(v_i),\mathcal S_{-i}(v_{-i}))\cap Path(\mathcal S_i(v_i'), \mathcal S_{-i}(v_{-i}'))$. 
    \item $v_i(f(v_i,v_{-i}))-P_i(v_i,v_{-i})< 
v_i(f(v_i',v_{-i}'))-P_i(v_i',v_{-i}')$.
\end{enumerate}
Then,  the strategy  $\mathcal S_i$ dictates the same message for the valuations  $v_i$ and $v_i'$ at vertex $u$.   
\end{lemma}

\section{Multi-Unit Auctions}\label{sec-mua}
We begin with the setting of multi-unit auctions. A multi-unit auction comprises of $m$ \emph{identical} items, where the valuation of every player $i$ is given by $v_i:[m]\to \mathbb R^{+}$. 
We consider two classes of valuations: unknown single-minded bidders (\cref{subsec::mua-sm}) and bidders whose valuations exhibit decreasing marginal values (\cref{subsec::mua-decreasing}).

\subsection{Single-Minded Valuations} \label{subsec::mua-sm}
In this section, we address two types of single-minded bidders: known and unknown. We begin with some definitions and then provide background on each.  

A valuation $v_i$ is \emph{single-minded} if there is a scalar $x_i$ and a quantity $d_i$ such that $v_i(q)=x_i$
if $q\ge d_i$, and otherwise $v_i(q)=0$.
In particular, we investigate the setting of \emph{unknown single-minded bidders}, where both the demand $d_i$ and the value $x_i$ of every player $i$ are private information (if only $x_i$ is private, then it is a setting with \emph{known single-minded bidders}).

For deterministic obviously strategy-proof mechanisms, the class of  known single-minded bidders admit strategy-proof mechanisms that give ${\mathcal O}(\min\{\log m,\log n\})$ approximation to the welfare \cite{DGR14,CGS22}, and no mechanism  gives an approximation better than $\Omega(\sqrt{\log n})$ of the welfare \cite{FPV21}.  For unknown single-minded bidders, \cite{Ron24} has shown a tight lower bound of $\min\{m,n\}$. 

We note that our proposed mechanisms (both the first attempt in \cref{subsub::first-attempt} and the actual one in \cref{subsub::actual-ub-mua-sm}) are for both unknown and known obviously strategy-proof mechanisms, whilst the impossibility that we provide in \cref{subsec-lbs-22-sm-mua} holds solely for unknown single-minded bidders. We leave open the question of understanding the approximation power of randomized obviously strategy-proof mechanisms for known single-minded bidders.


\subsubsection{Upper Bound: a First Attempt}\label{subsub::first-attempt}
Since an ascending-price auction for the grand bundle of goods is the optimal deterministic OSP mechanism in this setting \cite{Ron24}, we begin with a natural randomized analogue of this approach.  Namely, we ``guess'' a bundle size and run an ascending price auction for bundles of this size.  

Formally,
let $k = \lceil \log{m} \rceil$.  Consider the mechanism \textsc{Random-Bundles} in which 
an integer bundle size $\ell =2^{j}$ is sampled uniformly at random from the set of powers of two $\ell \in \{1,2,4,8,\dots,\frac{m}{2}, m\}$. Given this fixed bundle size, every bidder either wins exactly $\ell$ items or wins no items at all.  Observe, then, that at most $\nicefrac{m}{\ell}$ bidders can win a bundle. We now increase the price of being served $\ell$ items, until at most $\nicefrac{m}{\ell}$ bidders remain. All these bidders win $\ell$ items and pay the price at which we stop, while the remainder get nothing and pay nothing.
Note that \textsc{Random-Bundles} is a generalized ascending auction, so by \cref{lemma-partial} it is OSP.


We now argue that this mechanism obtains a logarithmic approximation to the optimal social welfare.  Due to space constraints, we provide a proof sketch of the approximation guarantee and defer the complete proof 
\cref{app-missing-mua}.

\begin{theorem}\label{claim:ascending-bundles-approx}
{\textsc{Random-Bundles}} obtains an $\mathcal O(\log{m})$ approximation to the optimal social welfare.
\end{theorem}
\begin{proof}[Proof Sketch]
First, observe that we can partition bidders into groups depending on their demand as follows:  for each bidder $i$, we place bidder $i$ in group $p$ if her demand  $d_i$ is between $2^{p}$ and $2^{p+1}-1$.
Since each bidder has demand at least $1$ and at most $m$, there are at most $\log{m}$ groups in total.  

Now, we  compare the portion of the optimal social welfare coming from bidders in group $p$ against the welfare \textsc{Random-Bundles} obtains when
 selecting bundles of size $2^{p+1}$.  On one hand, the optimal solution selects at most twice as many bidders appearing in group $p$ as the total number of bidders \textsc{Random-Bundles} serves conditioned on it selecting bundles of size $2^{p+1}$.  On the other hand, the bidders served in \textsc{Random-Bundles} conditioned on selecting bundles of size $2^{p+1}$ have the \emph{highest} value among bidders satisfied by receiving $2^{p+1}$ goods.  In total, we obtain an $\mathcal O(\log{m})$-approximation.
\end{proof}

\subsubsection{A Constant Upper Bound }\label{subsub::actual-ub-mua-sm}
Unfortunately, an approximation of $\mathcal O(\log m)$ appears to be the best achievable using the approach of randomly choosing fixed bundle sizes.  As such, we need to turn to a different approach.  We, thus, adopt the ``balanced sampling'' approach utilized extensively in other areas of mechanism design (see, e.g., \cite{feige2005competitive,goldberg2006competitive,dobzinski2012truthful,dobzinski2007two,badanidiyuru2012learning,bei2017worst})  in the form of Mechanism \ref{alg:single-minded},  below:

\begin{theorem}\label{thm-ub-mua-sm}
        There is a universally OSP mechanism for unknown single-minded bidders in a multi-unit auction  that gives a $400$-approximation to the optimal welfare.
\end{theorem}
We prove \cref{thm-ub-mua-sm} by describing a randomized mechanism, i.e., Mechanism \ref{alg:single-minded}, and showing it is universally OSP (\cref{claim-mua-sm-mechanism-osp}) and indeed gives a $400$-approximation to the optimal social welfare (\cref{lem:single-minded-approx}).

\begin{algorithm2e}
\SetKwInOut{Input}{Input}
\Input{A set of bidders $\bidders$ and $m$ identical items}
 With probability $\nicefrac{1}{2}$: 

 \quad Bundle all $m$ items together and run an ascending price auction on the grand bundle

 With remaining probability $\nicefrac{1}{2}$:

 \quad Let $S \leftarrow \emptyset$, $U \leftarrow \emptyset$
 
 \quad Place each bidder independently in $S$ w.p. $\nicefrac{1}{2}$ and each bidder in $U$ with the remaining probability

 \quad ``Discard'' each bidder and $S$ and learn their valuation function

 \quad Compute the optimal solution among bidders only in $S$ and let $O$ denote the value of this solution


\quad  Iterate over the bidders (in an arbitrary) order, and for each bidder $i \in \bidders$, let them purchase their preferred bundle from the remaining items at a price of $\nicefrac{O}{10m}$ per item
 
 
 \caption{``\textsc{Single-Minded}''}
 \label{alg:single-minded}
\end{algorithm2e}

\begin{lemma}\label{claim-mua-sm-mechanism-osp}
    Mechanism \ref{alg:single-minded} is universally OSP.
\end{lemma}
\begin{proof}
Under the realization of randomness where we auction the grand bundle, we utilize a generalized ascending auction for the grand bundle which is OSP by \cref{lemma-partial}.  
If we run the uniform pricing auction, then no bidder in $S$ can win any items and, thus, reporting their valuations truthfully is an obviously dominant strategy. Bidders in $U$ select their preferred bundle of goods, so reporting their preferences truthfully is also an obviously dominant strategy for them.
\end{proof}


To prove the approximation factor of Mechanism \ref{alg:single-minded}, 
we define a bidder as \emph{critical} if her value for the grand bundle of goods is at least $\nicefrac{1}{100}$ of the total optimal social welfare. We also use the notation $OPT$ to denote the optimal social welfare for a given valuation profile $(v_1, \ldots, v_n)$ and define $OPT(S)$ as the optimal social welfare attainable by allocating all items exclusively among the bidders in a specified subset $S \subseteq N$.

\cref{lem:auction-sampling} establishes that the sampling phase yields 
\textquote{representative}
sampled and unsampled sets in the case that there are no critical bidders with \textquote{high enough} probability. Note that the proof utilizes a lemma of \cite{bei2017worst}. 
We defer the proof of \cref{lem:auction-sampling}  to 
\cref{subsec::proof-auction-sampling}. 

\begin{lemma}\label{lem:auction-sampling}
Consider an instance $(v_1,\ldots,v_n)$ of bidders
in a multi-unit auction\footnote{
Our lemma actually holds also for bidders in a combinatorial auctions,  
but for simplicity we state it solely for multi-unit auctions.} where no bidder is critical.  Suppose each bidder is placed in a ``sampled'' set $S$ with probability $\nicefrac 1 2$ and placed in an ``unsampled'' set $U$ with the remaining probability independently. Then the optimal welfare obtained from bidders in the sampled set $\text{OPT}(S)$ and the optimal welfare obtained from bidders in the unsampled set $\text{OPT}(U)$ are such that $\text{OPT}(S) \geq \nicefrac{\text{OPT}}{5}$ and $\text{OPT}(U) \geq \nicefrac{\text{OPT}}{5}$ with probability at least $\nicefrac{1}{2}$.
\end{lemma}

With Lemma \ref{lem:auction-sampling} in hand, we are ready to prove that: 

\begin{lemma}\label{lem:single-minded-approx}
Mechanism \ref{alg:single-minded} obtains 
a $400$-approximation
to the optimal social welfare.
\end{lemma}
\begin{proof}
First we handle the case that there is a critical bidder, i.e., 
a bidder whose value for the grand bundle is at least $\nicefrac{OPT}{100}$. 
The existence of a critical bidder $i$ implies that allocating $i$ the grand bundle gives a $\nicefrac{1}{100}$-approximation to the optimal welfare.  Since we run an ascending auction on the grand bundle with probability $\nicefrac{1}{2}$ we obtain a $\nicefrac{1}{200}$-approximation  
in this case.

We now turn to the case that there does not exist a critical bidder.
In this case, Lemma \ref{lem:auction-sampling} implies that with probability $\nicefrac{1}{2}$ over the random sampling of bidders,
the optimal welfare achievable by the sampled set is within a factor $5$ of the optimal welfare.
As such when we proceed to the pricing phase, we set a price per item $p \in [\nicefrac{\text{OPT}}{50m}, \nicefrac{\text{OPT}}{10m}]$. Since we run a uniform price auction with probability $\nicefrac{1}{2}$, these conditions hold simultaneously with probability at least $\nicefrac{1}{4}$.
We perform case analysis on the number of items sold during this phase. 

Suppose the uniform pricing phase sells at least $\nicefrac{m}{2}$ goods.  In this case, since an unsampled bidder buying $t$ goods spends at least $\frac{t\text{OPT}}{50m}$, their value for the purchased bundle is at least $\frac{t\text{OPT}}{50m}$.  
Then, the total value of all bidders who purchase goods is at least $\frac{m}{2}\cdot\frac{\text{OPT}}{50m} = \frac{\text{OPT}}{100}$. Altogether, since we run uniform sampling with probability $\nicefrac{1}{2}$ and the estimation is
\textquote{good} with probability $\nicefrac{1}{2}$, we obtain  a $400$-approximation to the welfare.

Now, we analyze the complementary case 
where the uniform pricing phase sells fewer than $\nicefrac{m}{2}$ goods. For that, let $\vec{q}=(q_1,\dots,q_n)$ be the optimal allocation if the items are divided only among the bidders in $U$ (clearly, every bidder not in $U$ is allocated zero items).  

For that, observe that if  a bidder is allocated in $\vec q$ but not allocated in the allocation of the algorithm, it necessarily happens because of one of the following reasons. The first possibility is that the bidder is \emph{blocked}, meaning that   the number of items that she wants $d_i$ is not available when it is her turn. The other reason is that the bidder is \emph{small}, meaning that  $v_i(d_i)\le p\cdot d_i$, i.e., 
$d_i$ items are available  but
the 
price set is too high for her. 
Note that every bidder $i$ that 
is neither blocked nor small, is also satisfied in the algorithm.  

We will bound the loss of welfare from both kinds of bidders conditioned on our assumptions
of running a uniform price auction and having a ``balanced''  sampling  (i.e., both $OPT(S) \geq OPT/5$ and $OPT(U) \geq OPT/5$).
First, we show that blocked bidders do not exist, so they do not cause any loss of welfare. 
For that, we remind that by assumption the uniform phase sells less than $\frac
m 2$ items. Thus, a  blocked bidder wants to 
purchase strictly more than $\frac m 2
$ goods at a price of at least $\frac{\text{OPT}}{50m}$ and thus has a value of at least $\frac{\text{OPT}}{100}$, meaning that she is critical.  By assumption, there are no critical bidders, which implies that there are no blocked bidders, so they incur no loss of welfare.

We will now bound the welfare that comes from small bidders in $\vec q$. Note that the price $p$ that we set per item is at most $\nicefrac{OPT}{10m}$ and that the number of items allocated to small bidders in $\vec q$ is at most $m$. Since by definition $v_i(d_i)\le p\cdot d_i$, those bidders contribute to the welfare of $\vec q$ at most
$\frac{OPT}{10}$. 
Since the welfare of $\vec{q}$ is at least $\frac{OPT}{5}$,  bidders who are neither blocked or small contribute at least $\frac{OPT}{10}$ to the welfare of $\vec q$. Since 
Mechanism \ref{alg:single-minded} allocates to these bidders their desired  number of items, it achieves welfare of at least $\frac{OPT}{10}$. 
As we said before, this 
depends on finding a \textquote{good} partition of $U,S$ and running a uniform price auction which occurs in probability $\nicefrac{1}{4}$, 
so overall the expected welfare of the mechanism in this case is at least $\frac{OPT}{40}$. 

Combining all cases, we conclude that the expected welfare of Mechanism \ref{alg:single-minded} is at least $\frac{OPT}{400}$, thereby completing the proof.
\end{proof}
We note that
Mechanism \ref{alg:single-minded} also achieves a constant-approximation for bidders with 
{decreasing marginal valuations}. We discuss this in greater detail in \cref{subsec::mua-decreasing}.


\subsubsection{Lower Bound}\label{subsec-lbs-22-sm-mua}

\begin{theorem}\label{thm-mua-sm-lb}
For a multi-unit auction with $m\ge 2$ items and $n \ge 2$ unknown single-minded bidders,
no randomized mechanism that satisfies OSP, individual rationality and no negative transfers has approximation better than $\nicefrac{6}{5}$.
\end{theorem}
We note that that the proof uses a distribution of valuations that is based on the 
construction of \cite{Ron24}.  
\begin{proof}
We assume the domain $V_i$ of each bidder consists of single-minded valuations with values in  $\{0,1,\ldots,k^4\}$, where $k$ is an arbitrarily large number. 
Our example has only two bidders,  but it can be  extended to any number of bidders by adding bidders with the all-zero valuation. 

To  use 
    our variant of Yao's principle, 
we define a distribution $\mathcal D$ of valuation profiles
    and show that no deterministic mechanism that satisfies
    obvious strategy-proofness, individual rationality and no negative transfers with respect to $V=V_1\times V_2$ has approximation better than $\frac{6}{5}$ in expectation over $\mathcal D$. 
    To define it, 
consider the following valuations:  
\[
\renewcommand{\arraystretch}{1}
\begin{aligned}
v_i^{\text{one}}(x) &= \begin{cases}
1 & x \geq 1,\\
0 & \text{else.}
\end{cases}
\quad
v_i^{\text{ONE}}(x) = \begin{cases}
k^2 + 1 & x \geq 1,\\
0 & \text{else.}
\end{cases}\\[4pt]
v_i^{\text{all}}(x) &= \begin{cases}
k^2 & x=m,\\
0 & \text{else.}
\end{cases}
\quad
v_i^{\text{ALL}}(x) = \begin{cases}
k^4 & x=m,\\
0 & \text{else.}
\end{cases}
\end{aligned}
\]
Consider the following valuation profiles:
\[
\begin{aligned}
& I_1 = (v_1^{\text{one}}, v_2^{\text{one}}) \quad 
I_2 = (v_1^{\text{all}}, v_2^{\text{one}}) \quad 
I_3 = (v_1^{\text{ONE}}, v_2^{\text{ALL}}) \quad 
I_4 = (v_1^{\text{one}}, v_2^{\text{all}}) \quad
I_5 = (v_1^{\text{ALL}}, v_2^{\text{ONE}})
\end{aligned}
\]
Denote with $\mathcal D$ the distribution over 
profiles where the probability of  $I_1$ is $\frac{1}{3}$, and the probability of $I_2,I_3,I_4$ and $I_5$ is $\frac{1}{6}$ each. 
Observe that:
\begin{claim}\label{claim-mua-sm-instances}
Every deterministic mechanism that has approximation strictly better than $\nicefrac{6}{5}$ necessarily satisfies all of the following conditions:
\begin{enumerate}
    \item Given the valuation profile $I_1=(v_1^{one},v_2^{one})$, the mechanism
    allocates at least one item to every bidder. \label{condi-1}
    \item Given the valuation profile $I_2 = (v_1^{\text{all}}, v_2^{\text{one}})$, the mechanism allocates all items to bidder $1$.
    \label{condi-3}
\item Given the valuation profile $I_3 = (v_1^{\text{ONE}}, v_2^{\text{ALL}})$, the mechanism allocates all items to bidder $2$. \label{condi-2}
\item Given the valuation profile $I_4 = (v_1^{\text{one}}, v_2^{\text{all}})$, the mechanism allocates all items to bidder $2$. 
\item Given the valuation profile $I_5 =(v_1^{\text{ALL}}, v_2^{\text{ONE}})$, the mechanism allocates all items to bidder $1$. \label{condi-5}
\end{enumerate}
\end{claim}
The proof of \cref{claim-mua-sm-instances} is straightforward: if a deterministic mechanism violates one of the conditions, then since $k$ is arbitrarily large, then it  extracts at most $\frac{5}{6}$ of the optimal welfare in expectation over the distribution $\mathcal D$.

Fix a deterministic mechanism $A$ and strategies
$(\mathcal S_1,\mathcal S_2)$ that are individually rational and satisfy no negative transfers with respect to the valuations $V_1\times V_2$ 
and give approximation better than $\frac{6}{5}$
in expectation over the valuation profiles in the distribution $\mathcal D$. 
Let  $(f,P_1,P_2)$ be the allocation and payment rules that the mechanism $A$ and the strategies $(\mathcal S_1,\mathcal S_2)$ jointly realize. Assume towards a contradiction that $A$ and $(\mathcal S_1,\mathcal S_2)$ are OSP.

To analyze the mechanism, we focus on the following subsets of the domains of the valuations:
$
\mathcal{V}_1=\{v_1^{one},v_1^{ONE},v_1^{ALL}\}$ and  $\mathcal{V}_2=\{v_2^{one},v_2^{ONE},v_2^{ALL}\}$.\footnote{The cautious reader may have noticed that $\mathcal V_i$ does not contain  $v_i^{all}$. This is intentional, and it will be clear from the remainder of the proof why including this valuation is not necessary.} 
We begin by observing that there necessarily exists a vertex $u$, and valuations $v_1,v_1' \in \mathcal{V}_1$, and  $v_2,v_2' \in \mathcal{V}_2$ such that $(\mathcal{S}_1(v_1), \mathcal{S}_2(v_2))$ diverge at vertex $u$. This follows from \cref{claim-mua-sm-instances}, which implies that the mechanism $A$ must output different allocations for different valuation profiles in $\mathcal{V}_1 \times \mathcal{V}_2$. Consequently, not all valuation profiles end up in the same leaf, meaning that divergence must occur at some point. 
 
 Let $u$ be the first vertex in the protocol such that 
 $(\mathcal{S}_1(v_1),\mathcal{S}_2(v_2))$ and $(\mathcal{S}_1(v_1'),\mathcal{S}_2(v_2'))$ diverge, i.e., dictate different messages. 
Note that by definition this implies that $u\in Path(\mathcal{S}_1(v_1),\mathcal{S}_2(v_2))\cap Path(\mathcal{S}_1(v_1'),\mathcal{S}_2(v_2'))$ and that either bidder $1$ or bidder $2$ sends different messages for the valuations in $\mathcal{V}_1$ or $\mathcal V_2$, respectively. 
Without loss of generality, we assume that bidder $1$ sends different messages, meaning that there exist $v_1,v_1'\in \mathcal{V}_1$ such that $\mathcal S_1(v_1)$ and $\mathcal S_1(v_1')$ dictate different messages at vertex $u$.  
We remind that $\mathcal{V}_1=\{v_1^{one},v_1^{ONE},v_1^{all}\}$, so the  following claims jointly imply a contradiction, completing the proof: 
\begin{claim}\label{claim-oneone-same}
    The strategy $\mathcal S_1$ dictates the same message at vertex $u$ for the valuations $v_1^{one}$ and $v_1^{ONE}$. 
\end{claim}
\begin{claim}\label{claim-ONE-ALL-same}
        The strategy $\mathcal S_1$ dictates the same message at vertex $u$ for the valuations $v_1^{ONE}$ and $v_1^{ALL}$.
\end{claim}
We include the proofs for the sake of completeness, {but note that} they are identical to the proofs provided in \cite{Ron24}.
The proofs make use of the following lemma, which is a collection of observations about the allocation and the payment scheme of player $1$:    
\begin{lemma}\label{lemma-small-pay}
    The allocation rule $f$ and the payment scheme $P_1$ of bidder $1$ satisfy that:
    \begin{enumerate}
        \item Given $(v_1^{one},v_{2}^{one})$, bidder $1$ wins at least one item and pays at most $1$.  \label{item-1}
        \item  Given $(v_1^{ONE},v_2^{ALL})$, bidder $1$ gets the empty bundle and pays zero.   \label{item-2}
        \item Given $(v_1^{ALL},v_2^{one})$, bidder $1$ wins all the items and pays at most $k^2$. \label{item-3}  
    \end{enumerate}
\end{lemma}
The lemma is a direct consequence of the approximation guarantees of the mechanism, together with the fact that it is obviously strategy-proof and satisfies individual rationality and no negative transfers.  We use \cref{lemma-small-pay} 
 now and defer the proof to \cref{subsec::proof-lemma-small-pay}.
\begin{proof}[Proof of \cref{claim-oneone-same}]
     Note that by Lemma \ref{lemma-small-pay} item \ref{item-1}, $f(v_1^{one},v_2^{one})$ allocates at least one item to player $1$ and $P_1(v_1^{one},v_2^{one})\le 1$. Therefore:
\begin{equation}\label{eq-good-leaf1}
 v_1^{ONE}(f(v_1^{one},v_2^{one}))-P_1(v_1^{one},v_2^{one})\ge k^2   
\end{equation}
 In contrast, by part \ref{item-2} of Lemma \ref{lemma-small-pay},   $f(v_1^{ONE},v_2^{ALL})$ allocates no items to player $1$ and $P_1(v_1^{ONE},v_2^{ALL})=0$, so:
 \begin{equation}\label{eq-bad-leaf1}
 v_1^{ONE}(f(v_1^{ONE},v_2^{ALL}))-P_1(v_1^{ONE},v_2^{ALL})= 0   
\end{equation}
Combining inequalities (\ref{eq-good-leaf1}) and (\ref{eq-bad-leaf1}) gives:
\begin{align*}
  v_1^{ONE}(f(v_1^{ONE},v_2^{ALL}))-P_1(v_1^{ONE},v_2^{ALL})< 
  v_1^{ONE}(f(v_1^{one},v_2^{one}))-P_1(v_1^{one},v_2^{one})  
\end{align*}
We remind that vertex $u$ belongs in $Path(\mathcal S_1(v_1^{one}),\mathcal S_2(v_2^{one}))$ and also in
$Path(\mathcal{S}_1(v_1^{ONE}),\allowbreak\mathcal{S}_2(v_2^{ALL}))$. Therefore, Lemma \ref{lemma-bad-leaf-good-leaf} gives that the strategy $\mathcal S_1$ dictates the same message for  $v_1^{one}$ and $v_1^{ONE}$ at vertex $u$. 
\end{proof}
\begin{proof}[Proof of \cref{claim-ONE-ALL-same}]
    Following the same approach as in the proof of Claim \ref{claim-oneone-same}, note that by \cref{lemma-small-pay} \cref{item-3}: 
\begin{equation}\label{break-align}
v_1^{ONE}(f(v_1^{ALL},v_2^{one}))-P_1(v_1^{ALL},v_2^{one}) \ge k^2+1-k^2     
\end{equation}
Also, by \cref{lemma-small-pay}  \cref{item-2}:
\begin{equation}\label{break-align2}
  v_1^{ONE}(f(v_1^{ONE},v_2^{ALL}))  
-P_1(v_1^{ONE},v_2^{ALL})=0   
\end{equation}
Combining \cref{break-align} and \cref{break-align2}:
\begin{equation*}
    v_1^{ONE}(f(v_1^{ONE},v_2^{ALL}))  
-P_1(v_1^{ONE},v_2^{ALL})<  v_1^{ONE}(f(v_1^{ALL},v_2^{one}))-P_1(v_1^{ALL},v_2^{one})
\end{equation*}
Given the above inequality with the fact that 
vertex $u$ belongs in $Path(\mathcal S_1(v_1^{ALL}),\mathcal S_2(v_2^{one}))$ and in
$Path(\mathcal{S}_1(v_1^{ONE}),\mathcal{S}_2(v_2^{ALL}))$,
Lemma \ref{lemma-bad-leaf-good-leaf}
implies  
that the strategy $\mathcal S_1$ dictates the same message for the valuations $v_1^{ONE}$ and $v_1^{ALL}$ at vertex $u$. 
\end{proof}

\end{proof}

\subsection{Decreasing Marginal Valuations}\label{subsec::mua-decreasing}
In this section, we consider valuations that exhibit decreasing marginals. 
A valuation $v:[m]\to \mathbb R$  has decreasing marginals if  for every quantity $j\in [m]$, $v(j)-v(j-1)\ge v(j+1)-v(j)$. 
This is the only class of multi-parameter valuations for which the power of deterministic obviously strategy-proof mechanisms is not yet understood: 
the best-to-date mechanism achieves an $\mathcal O(\log n)$ approximation, and no mechanism for two bidders and two items obtains approximation better than $\sqrt 2$ \cite{GMR17}. We begin by showing that if we allow randomization:
\begin{theorem}\label{thm:decreasing-marginals}
There exists a randomized obviously strategy-proof mechanism that achieves a $400$\allowbreak-approximation to the optimal social welfare for bidders with decreasing marginal valuations.
\end{theorem}
We note that the mechanism described in the proof of \cref{thm:decreasing-marginals} corresponds to Mechanism \ref{alg:single-minded}, which was previously introduced for the class of single-minded bidders. The proof of \cref{thm:decreasing-marginals} is deferred to \cref{subsec::proof--dec-mua}.

Having established a constant-factor randomized obviously strategy-proof mechanism for decreasing marginal valuations, a natural question arises: is this result tight? Specifically, can we establish impossibility results for this class? To further deepen our understanding of this class of valuations, we now describe a phenomenon that highlights the challenges in proving such impossibilities.

\subsubsection{A Non-Monotonicity Effect for Bidders with Decreasing Marginal Values}\label{subsub::non-mono}
In this section, we describe a non-monotonicity phenomenon that occurs for obviously strategy-proof mechanisms in multi-unit auction with decreasing marginal valuations.  In contrast to the rest of the paper, we focus on deterministic mechanisms rather than randomized ones. 
The phenomenon is that    \emph{deterministic} obviously strategy-proof mechanisms for bidders with decreasing marginal valuations, 
 adding an item improves the approximation power:
\begin{theorem}\label{thm-lb-mua-dec}
    For $2$ bidders and $2$ items, no obviously strategy-proof \emph{deterministic} mechanism that satisfies individual rationality and no negative transfers gives approximation better than $2$.  
\end{theorem}

\begin{lemma}\label{lemma:mono-mua-dec}
        For $2$ bidders and $3$ items, there is an obviously strategy-proof 
\emph{deterministic} mechanism that gives an approximation of  $1.5$.
\end{lemma}

Observe that this is not typical, as the approximation guarantee of mechanisms typically deteriorates as the number of items increases: intuitively, the more items there are, the \textquote{harder} it becomes to allocate them optimally. From a more formal perspective, impossibility results for auctions with $m$ items can be extended to those with $m+1$ items by introducing a \textquote{dummy} item that no bidder values. However, for the class of valuations with decreasing marginals, the additional item enables a new mechanism: we can now 
 allocate one item to each bidder and then run an ascending auction on the remaining item. 


Note that the previously known lower bound on deterministic obviously strategy-proof mechanisms is $\sqrt 2$ \cite{GMR17}.\footnote{Note that \cite{GMR17} actually prove an impossibility for all mechanisms that are weakly group strategy-proof. This implies an impossibility for obviously strategy-proof mechanisms because as \cite{li} shows, every obviously strategy-proof mechanism is weakly group strategy-proof.} However, in contrast to \cite{GMR17}, \cref{thm-lb-mua-dec} applies solely to mechanisms that satisfy individual rationality and no negative transfers.   
The proofs  of \cref{thm-lb-mua-dec} and \cref{lemma:mono-mua-dec} can be found in \cref{subsec-lb-proof-mua-dec,sec-impos-mua-dec}.  


\section{Combinatorial Auctions}\label{sec-combi}
We now turn to settings with heterogeneous items.
We explore settings involving additive and unit-demand bidders and conclude by considering mechanisms for subadditive and general valuations.
The proofs of all theorems can be found in 
\cref{app-missing-combi}. 
\subsection{Additive Valuations}
A valuation $v_i$ is \emph{additive} if bidder $i$ has a value $v_{ij} \geq 0$ for item $j$ and the value bidder $i$ has for receiving a set of items $A_i$ is equal to $\sum_{j \in A_i}{v_{ij}}$. 
\subsubsection{Upper Bound}
We show that the sampling approach yields a $4$-approximation for this setting:
\begin{theorem}\label{thm-ub-add}
    There is a universally OSP mechanism for bidders with additive valuations that gives a $4$-approximation to the optimal welfare.
\end{theorem}
We now describe Mechanism \ref{alg:additive}. Simply put, Mechanism \ref{alg:additive} samples a threshold price for each item and then uses these threshold prices as a posted-price mechanism for the unsampled bidders. To handle tie-breaking, priority is given to bidders with higher indices. 
\begin{algorithm2e}
\SetKwInOut{Input}{Input}
\Input{A set of bidders $\bidders$ and a set of $M$ items}

 Index the bidders in some arbitrary fixed order

 $S \leftarrow \emptyset$, $U \leftarrow \emptyset$
 
Independently assign each bidder to set $S$ with probability $\nicefrac{1}{2}$ and to set $U$ with probability $\nicefrac{1}{2}$.

 ``Discard'' each bidder in $S$ and learn their value for each item

 For each $j \in M$: set a price $p_j$ on item $j$ equal to $\max_{i \in S}{v_{ij}}$ and let $n(j)$ denote the smallest index among bidders in $\arg\max_{i \in S}{v_{ij}}$
 
 For each $i \in U$ in an arbitrary order:  Let $i$ purchase all previously unsold items $j \in M$ for which either: (i) $v_{ij} > p_j$; or (ii) $v_{ij} = p_j$ and $i$ has a lower index than $n(j)$
 
 \caption{``\textsc{Additive}''}
 \label{alg:additive}
\end{algorithm2e}
The following two lemmata jointly provide the proof of \cref{thm-ub-add}.
\begin{lemma}\label{lemma-add-osp}
Mechanism \ref{alg:additive} is universally OSP.
\end{lemma}
The proof of \cref{lemma-add-osp} is straightforward, and we write it for completeness in \cref{subsec::proof-add-osp}.

\begin{lemma}\label{lem-add-approx}
Mechanism \ref{alg:additive} obtains a $4$-approximation to the optimal social welfare in the presence of additive bidders.
\end{lemma}
\begin{proof}
Observe that since the valuation functions are additive,  optimal solutions must allocate each item $j$ to some bidder $i \in \text{argmax}_{i \in \bidders}{v_{ij}}$.  In particular, one optimal solution allocates each item $j$ to the bidder $i^*_j \in \text{argmax}_{i \in \bidders}{v_{ij}}$ with the smallest index according to the order over the bidders that the mechanism specifies.  We argue that each item $j$ is allocated to its corresponding bidder $i^*_j$ by Mechanism \ref{alg:additive} with probability at least $\frac 1 4$. Linearity of expectation directly implies a $4$-approximation to the optimal welfare.

To see that we allocate each item $j$ to
$i^*_j$
with probability $\frac 1 4$, first observe that $i^*_j$ is placed in $U$ with probability $\frac 1 2$.  Moreover, let $\tilde{i}_j$ denote the \textquote{second-highest bidder}, which we define as
the bidder in $\text{argmax}_{i \in \bidders \setminus \{i^*_j\}}{v_{ij}}$
that has the smallest index.
 This bidder is placed in $S$ (independently of the placement of $i^*_j$) with probability $1/2$.  
Observe that when bidder $\tilde{i}_j$ is in $S$ and $i^*_j$ is in $U$,
then item $j$ is necessarily available for bidder $i^*_j$ who gets it indeed. 
This occurs with probability at least $\frac{1}{4}$, as desired.\qedhere
\end{proof}




\subsubsection{Lower Bound}
\begin{theorem}\label{thm-lb-add}
For a combinatorial auction with $m\ge 2$ items and $n\ge 2$ additive bidders, there 
is no randomized obviously strategy-proof mechanism that satisfies individual rationality and no negative transfers and gives approximation better than $\frac{8}{7}$ to the optimal social welfare.
\end{theorem}
The proof follows the same structure as the proof of \cref{thm-mua-sm-lb} of describing 
a distribution $\mathcal D$ and showing that it is hard for every deterministic mechanism. By applying Yao's Lemma (\cref{lem:yaos}), we get hardness for randomized obviously strategy-proof mechanisms. 
The proof can be found in \cref{lb-add-proof-place}. 
We note that the construction of \cref{thm-lb-add} is similar to a construction in \cite{Ron24}.  However, the case analysis we employ is more involved as it includes additional valuation profiles.

\subsection{Unit-Demand Valuations}
We now address bidders with \emph{unit-demand} valuations, where there is a value $v_{ij} \geq 0$ for each $i \in \bidders$ and  $j \in \items$ and the value bidder $i$ has for a set $A_i$ is equal to $\max_{j \in A_i}{v_{ij}}$.

\subsubsection{Upper Bound}
This setting appears more complicated than the setting of additive valuations, as the approach of setting the price of each item to be the price of the second highest bid fails miserably:
\begin{example}\label{ex:ud-failure}
    Consider the instance where there are $\sqrt{n}$ \textquote{high} bidders with value $2$ for all items, and the rest of the bidders are \textquote{low}, in the sense that they value all items at $1$. Assume that the number of items, $m$, is equal to the number of bidders, $n$.
    
    Note that if we sample roughly half of the bidders and use their highest values to determine the prices for the unsampled bidders, as we
    do in Mechanism \ref{alg:additive},
    we get $\approx\frac{1}{\sqrt n}$ of the optimal welfare. This is because at least one of the \textquote{high} bidders is in the sample with probability $1-\frac{1}{2^{\sqrt n}}$. Thus in this very likely case, the price of all items is set to be $2$ and none of the \textquote{low} bidders take any item, so the welfare obtained in expectation is at most $\sqrt{n}$. However, the  optimal welfare is $n+\sqrt{n}$. 
\end{example}
Thus, to obtain a constant factor approximation for this setting,  we use the beautiful algorithm of \cite{reiffenhauser2019optimal}, originally formulated for the problem of strategy-proof online matching: 
\begin{algorithm2e}
\SetKwInOut{Input}{Input}
\Input{A set of bidders $\bidders$ and a set of  items $M$}

Ensure uniqueness of all optimal solutions for any fixed subset of bidders and items by fixing a consistent tie-breaking rule between optimal allocations

 Choose uniformly at random permutation $\sigma$ over the bidders and index the bidders in this order

 $S \leftarrow \emptyset$, $M_A \leftarrow M$
  
 ``Discard'' the first $\lfloor n/e \rfloor$ bidders and learn their value for each item and add these bidders to $S$

 For each consecutive bidder $i \in \{\lfloor n/e \rfloor + 1, \dots, n\}$:

    \quad Compute a price $p_j$ for each item $j \in M_A$ equal to $OPT(S, M_A) - OPT(S, M_A \setminus \{j\})$ (i.e., the decrease in welfare if $j$ were taken away from bidders in $S$)

    \quad Let $i$ purchase her favorite item (i.e., the item for which $v_{ij} - p_j$ is maximized and greater than $0$) at the current prices and let $j^i$ denote this item (if any). Use the tie breaking rule to determine whether player $i$ can take items for which $v_{ij}-p_j=0$

    \quad $M_A \leftarrow M_A \setminus \{j^i\}$, $S \leftarrow S \cup \{i\}$

    \quad Ask $i$ for her value of $i$ for each item
 
 \caption{``\textsc{Unit-Demand}'' (adapted from Algorithm 2 of \cite{reiffenhauser2019optimal})}
 \label{alg:unit-demand}
\end{algorithm2e}

We note that Mechanism \ref{alg:unit-demand} is Algorithm 2 of \cite{reiffenhauser2019optimal}, which we slightly adapt to our offline setting and rephrase to make the fact that the mechanism is universally obviously strategy-proof more clear.

\begin{theorem}\label{thm:ud-upper}
Mechanism \ref{alg:unit-demand} is universally OSP and achieves an $e$-approximation
to the optimal social welfare in the presence of unit-demand bidders.
\end{theorem}
\begin{proof}
The approximation ratio of the mechanism directly follows from Theorem 1 of \cite{reiffenhauser2019optimal}.  To see that the mechanism is obviously strategy-proof, observe that the discarded bidders obtain no utility regardless of their report (and, thus, true value reporting is weakly obviously dominant). As for the remaining bidders, they get to select their most preferred remaining item. After that, their reported information does not affect their utility. Therefore, picking their favorite remaining item and then answer the queries afterwards truthfully is an obviously dominant strategy.   
\end{proof}



\subsubsection{Lower Bound}
\begin{theorem}\label{thm-lb-ud}
For a combinatorial auction with $m\ge 2$ items and $n\ge 2$ unit-demand bidders, there 
is no randomized obviously strategy-proof mechanism that satisfies individual rationality and no negative transfers and gives approximation better than $\frac{8}{7}$ to the optimal social welfare.
\end{theorem}
Similarly to \cref{thm-mua-sm-lb} and \cref{thm-lb-add}, the proof follows the structure of describing a distribution $\mathcal{D}$, proving that it is hard for every deterministic mechanism, and applying Yao's lemma (\cref{lem:yaos}). In fact, the proof closely resembles that of \cref{thm-lb-add} due to the fact that most valuations used in both constructions are simultaneously additive and unit-demand. The full proof is provided in \cref{lb-ud-proof-place}.


\subsection{More General Valuations}
In light of our previous results, one may wonder whether there exists a \textquote{rich enough} class of valuations for which randomized OSP mechanisms are provably unable to extract more than a constant fraction of the welfare. Perhaps surprisingly, the answer is no. However, the state-of-the-art \emph{computationally efficient} randomized mechanisms for subadditive and general valuations\footnote{A function $v$ is subadditive if for every two bundles of items $A, B \subseteq M$, it holds that $v(A \cup B) \leq v(A) + v(B)$. A valuation is general monotone if for every two bundles $A\subseteq B$, it holds that $v(B)\ge v(A)$.} are, in fact, universally OSP:
\begin{claim}\label{cl:subadditive}
    The $O((\log\log(m))^3)$-approximate randomized mechanism of \cite{assadi2021improved} for subadditive valuations is universally OSP.
\end{claim}
\begin{claim}\label{cl:general}
The $O(\sqrt{m})$-approximate randomized mechanism of \cite{dobzinski2012truthful} for general valuations is universally OSP.
\end{claim}
The proofs are straightforward and can be found in \cref{subsec::proofs-subadd-general}. 

\bibliographystyle{alpha}
\bibliography{bibliography}

\clearpage

\appendix

\section{Additional Formalities}\label{app-formalities}
We now provide a formal definition of obviously dominant strategies. 
and note that our definitions and setup closely follows \cite{Ron24}.  Fixing a protocol and behavior $B_i$ of bidder $i$ we say that vertex $u$ of the protocol is \emph{attainable} if there exists some $B_{-i}\in \mathcal{B}_{-i}$ (i.e., some profile of behaviors for the other players) such that $u\in Path(B_i,B_{-i})$. 
We now may formally define an \emph{obviously dominant behavior}:
\begin{definition}[Definition 2.1 in \cite{Ron24}]\label{def-obvs-behavior}
Consider a deterministic mechanism $A$, together with a behavior $B_i$ and a valuation $v_i$ of some player $i$.
Fix a vertex $u\in \mathcal N_i$ that is attainable given the behavior $B_i$. 
Behavior $B_i$ is an \emph{obviously dominant behavior for player $i$ at vertex $u$ given the valuation $v_i$} if for every behavior profiles $B_{-i}\in \mathcal B_{-i}$ and  $(B_1',\ldots,B_n')\in \mathcal B_1 \times \cdots \times \mathcal{B}_n$ such that:
\begin{enumerate}
    \item $u\in Path(B_1,\ldots,B_n)\cap Path(B_1',\ldots,B_n')$ \emph{and} 
    \item $B_i$ and $B_i'$ dictate sending different messages at vertex $u$.
\end{enumerate}
it holds that:
$$
v_i(f_i(B_i,B_{-i})) - p_i(B_i,B_{-i}) \geq v_i(f_i(B_i',B_{-i}')) - p_i(B_i',B_{-i}')
$$
\end{definition}

Note that \cref{def-obvs-behavior} only deals with behaviors at individual nodes in the protocol.  We then say that a behavior $B_i$ is an obviously dominant behavior given valuation $v_i$ if it is an obviously dominant behavior for player $i$ at
all  attainable vertices.  Formally:
\begin{definition}[Definition 2.2 in \cite{Ron24}]
Fix a behavior $B_i$ together with the subset of vertices in $\mathcal N_i$ that are attainable for it,  which we denote with $U_{B_i}$. Fix a valuation $v_i$ of player $i$.
The behavior $B_i$ is an \emph{obviously dominant behavior for player $i$ given the valuation $v_i$} if
it is an obviously dominant behavior for player $i$ given the valuation $v_i$ for every vertex $u\in U_{B_i}$. 
\end{definition} 

Finally, the definition of \emph{obviously dominant strategies} naturally follows.  Particularly, an obviously dominant strategy is one that defines an obviously dominant behavior for each possible valuation.  Formally:
\begin{definition}[Definition 2.3 in \cite{Ron24}]
A strategy $\mathcal{S}_i$ of player $i$ is an \emph{obviously dominant strategy} if for every $v_i$, the behavior $\mathcal S_i(v_i)$ is an obviously dominant behavior.
\end{definition}

\paragraph{Weak Monotonicity}
Weak monotonicity is a property of social choice functions.  To define it, we denote with $f_i$ the function that outputs for every player $i$ the bundle that she wins given $f$.

    An allocation rule $f:V\to\allocs$ is \emph{weakly monotone} if for every player $i$, for every valuation profile $v_{-i}$ of the bidders $N\setminus \{i\}$, and every two valuations $v_i,v_i'\in V_i$, it holds that if $f_i(v_i,v_{-i})=S$ and $f_i(v_i',v_{-i})=S'$, then $v_i(S)-v_i(S')\ge v_i'(S)-v_i'(S')$. 
    It is well known that:
\begin{lemma}\cite{BCLMNS06,LMN03}\label{wmon-lemma}
    Every allocation rule that is implemented by a dominant-strategy mechanism is weakly monotone. 
\end{lemma}


\section{Improved Approximation Ratios For Two Agents and Two Items}
\label{sec-22}
In this section, we describe improved approximation guarantees for the special case of two bidders and two items. Observe that this is the simplest case for which the power of randomized obviously strategy-proof mechanisms is not resolved.\footnote{If we consider only one bidder, then a mechanism that always allocates to her all the items is optimal and obviously strategy-proof. If there is only one item, then an ascending auction on this item among the bidders is optimal and obviously strategy-proof.} We believe that understanding the power of randomized obviously strategy-proof mechanisms in this simple case is a first step towards understanding their power for the more general settings of arbitrary number of bidders and items. 

We consider three classes of valuations: multi-unit auctions with unknown single-minded bidders (\cref{subsec:22-mua-sm}), which admit a $\frac{4}{3}$-approximation, combinatorial auctions with subadditive bidders (\cref{subsec:22-ca-subadd}), which also admit a $\frac{4}{3}$-approximation, and combinatorial auctions with general bidders (\cref{subsec:22-ca-general}), which admit a $\frac{3}{2}$-approximation.

We remind the reader that for unknown single-minded bidders in a multi-unit auction, even with just two bidders and two items, no mechanism achieves an approximation better than $\frac{5}{6}$ (\cref{thm-mua-sm-lb}). Thus, there remains a gap of $\frac{1}{12}$ in our understanding of this setting.
Since unit-demand and additive valuations are subadditive, these classes also admit a $\frac{4}{3}$-approximation. As we have shown that even for two bidders and two items (\cref{thm-lb-ud,thm-lb-add}) no mechanism achieves an approximation better than $\frac{7}{8}$, there remains a gap of $\frac{1}{8}$.

In the proofs, we denote $v_i(j~|~j') = v_i(\{j,j'\}) - v(\{j'\})$ as the marginal gain in value of bidder $i$ when adding item $j$ to her bundle already containing item $j'$. We denote the optimal welfare with $\opt$. 

\subsection{Multi Unit Auction for Single-Minded Bidders with Unknown Demands} \label{subsec:22-mua-sm}
Let $p = 1/2$ and consider the mechanism $\mathcal{M}_1$ which with probability $p$ runs a second-price auction for the grand bundle and with probability $(1-p)$ allocates a uniformly random agent a single item and runs a second-price auction for the second item.  
\begin{claim}\label{claim-mec-sm-mua-osp}
   The mechanism $\mathcal{M}_1$ is universally OSP.  
\end{claim}
\begin{proof}
Observe that $\mathcal{M}_1$ is a probability distribution over two deterministic mechanisms, both of which are generalized ascending auctions. By \cref{lemma-partial}, they are therefore obviously strategy-proof. Thus, $\mathcal M_1$ is obviously strategy-proof in the universal sense.
\end{proof}

\begin{claim}
The mechanism $\mathcal{M}_1$ achieves a $\nicefrac 4 3$-approximation to the optimal social welfare.
\end{claim}
\begin{proof}
    Let $v_i$ denote the private value of agent $i$ if she receives a bundle which satiates her demand. 
    
    First suppose that the optimal allocation obtains positive value from exactly one agent.  Without loss of generality, say that this is agent $1$.  We may immediately conclude that $v_1 \geq v_2$.  Consider that, by definition, agent $1$ is satiated by the grand bundle.  Thus, with probability $p$, the mechanism $\mathcal{M}_1$ satiates agent $1$ and receives a value $v_1$.  On the other hand, with probability $(1-p)/2$, the mechanism $\mathcal{M}_1$ allocates agent $1$ the first item and runs a second-price auction for the second item.  This necessarily yields total welfare at least $v_1$ since $v_1 \geq v_2$ and since agent $1$ will either have marginal value for the second item equal to $v_1$ if she is satiated only by both items or marginal value $0$ if she is satiated by a single item.  Thus, in the case where the optimal allocation serves exactly one agent, $\mathcal{M}_1$ achieves an approximation ratio of at least $(p + (1-p)/2)$.

    Now suppose that the optimal allocation obtains positive value from both agents.  Note that we may immediately conclude that each agent is satiated when she receives at least one item.  Without loss of generality, suppose that $v_1 \geq v_2$.  But then, with probability $p$ $\mathcal{M}_1$ allocates both items to bidder $1$ and obtains welfare $p\cdot v_1$.  On the other hand, in the case that $\mathcal{M}_1$ allocates a single item to a uniformly random agent and runs a second-price auction for the second item, we have that the agent who is randomly given the first item has $0$ marginal value for the second item whereas the other agent has positive marginal value for the item.  But then, in this case, which occurs with probability $(1-p)$, $\mathcal{M}_1$ obtains the optimal welfare.  As such, the expected welfare of $\mathcal{M}_1$ is $p \cdot v_1 + (1-p)\cdot (v_1 + v_2)$.  Since $v_1 \geq v_2$ we then obtain an approximation of $p/2 + (1-p)$.

    The approximation in the first case is monotonically increasing in $p$ whereas the approximation in the second case is monotonically decreasing.  Taking $1/2 + p/2 = p/2 + (1-p)$ gives $p = 1/2$ and thus setting $p=\nicefrac{1}{2}$ implies that $\mathcal M_1$ an approximation ratio of $\nicefrac 4 3$ as desired.
\end{proof}

\subsection{Combinatorial Auction with Subadditive Bidders}\label{subsec:22-ca-subadd}
Consider the mechanism $\mathcal{M}_2$ which uniformly at random selects a bidder $i$ and uniformly at random selects an item $j$ and allocates item $j$ to bidder $i$ and after this allocation occurs runs an ascending second-price auction for the remaining item.  Note that $\mathcal M_2$ is in fact a probability distribution over generalized ascending auctions, so by \cref{lem:auction-sampling} it is OSP in the universal sense.
We argue that it achieves a $\frac 4 3$-approximation for the broad class of \emph{monotone subadditive} valuations.  A function
$v:2^M\to \mathbb{R}^{\geq 0}$ is \emph{monotone} if for all bundles $T \subseteq T' \subseteq M$ we have that $v(T) \geq v(T')$.\footnote{Monotonicity is sometimes called ``free-disposal'' in the literature since a bidder is always weakly happier to receive more goods (because she can always ``dispose'' of goods she is unhappy to receive).}  A function $v:2^M\to \mathbb{R}^{\geq 0}$
is \emph{subadditive}  if for all bundles $T, T' \subseteq M$ we have that $v(T) + v(T') \geq v(T \cup T')$. 

Subadditive valuations are more general than submodular valuations and monotone subadditive valuations capture, as special cases, both additive and unit-demand valuations.  As such, $\mathcal{M}_2$ provides improved approximations for two bidders and two items in combinatorial auctions with both unit-demand and additive valuations.
\begin{claim}
The mechanism $\mathcal{M}_2$ achieves a $\nicefrac 4 3$-approximation to the optimal social welfare for monotone subadditive valuations.
\end{claim}
\begin{proof}
    Let $\{1,2\}$ denote the set of two bidders and $\{a,b\}$ denote the set of two items.  
    We consider two cases depending on whether the optimal allocation is such that one bidder receives both items or  each bidder receives an item.

    We begin with the former case and without loss of generality assume that the optimal allocation awards both items to bidder $1$ (the case where bidder $2$ obtains both items is symmetric), i.e., $\opt = v_1(\{a,b\})$.   Since $\opt = v_1(\{a,b\})$, it must be that $v_1(a~|~b) \geq v_2(\{a\})$ and $v_1(b~|~a) \geq v_2(\{b\})$. 
    
    We now consider all possible outcomes of the randomized mechanism $\mathcal M_2$. If the mechanism $\mathcal M_2$ randomly allocates either item $a$ or $b$ to bidder $1$ our auction obtains welfare equal to $v_1(\{a,b\})$.  On the other hand, if we randomly allocate $a$ to bidder $2$ then our auction obtains welfare of $v_2(\{a\}) + \max\{v_2(b~|~a),v_1(\{b\})\} \geq v_1(\{b\})$ and, similarly, if we allocate $b$ to bidder $2$ then our auction obtains $v_2(\{b\}) + \max\{v_2(a~|~b),v_1(\{a\})\} \geq v_1(\{a\})$.  Combining these cases with the probabilities they occur gives that $\mathcal M_2$ achieves expected welfare at least $\frac{1}{2}\cdot v_1(\{a,b\}) + \frac{1}{4}\cdot v_1(\{a\}) + \frac{1}{4}\cdot v_1(\{b\})$.  But then, by subadditivity, we have that the auction obtains welfare at least $\frac{3}{4}\cdot v_1(\{a,b\})$.

    Now consider the latter case where the optimal allocation awards each bidder a single item and suppose that bidder $1$ receives $a$ and bidder $2$ receives $b$ (the opposite case is symmetric), i.e., $\opt = v_1(\{a\}) + v_2(\{b\})$.  We then have that $v_1(\{a\}) \geq v_2(\{a~|~b\})$ and $v_2(\{b\}) \geq v_1(\{b ~|~ a\})$.  This means that if $\mathcal M_2$  randomly allocates item $a$ to bidder $1$ or item $b$ to bidder $2$ our auction then obtains the optimal social welfare.  On the other hand, if we allocate $b$ to $1$ our auction obtains welfare $v_1(\{b\}) + \max\{v_1(a~|~b) ,v_2(\{a\})\} \geq v_1(\{a,b\}) \geq v_1(\{a\})$ by monotonicity and similarly if we allocate $a$ to $2$ our auction obtains welfare $v_2(\{a\} + \max\{v_2(b~|~a) ,v_1(\{b\})\} \geq v_2(\{a,b\}) \geq v_2(\{b\})$.  
    
    Combining these cases with the probabilities they occur implies that the expected welfare of $\mathcal M_2$ is at least $\frac{1}{2}\left(v_1(\{a\}) + v_2(\{b\})\right) + \frac{1}{4}v_1(\{a\}) + \frac{1}{4}v_2(\{b\}) = \frac{3}{4}\left(v_1(\{a\}) + v_2(\{b\})\right)$, as desired.
\end{proof}

\subsection{Combinatorial Auction with General Monotone Bidders}\label{subsec:22-ca-general}
Let $p$ be equal to $\frac{1}{3}$. 
Consider the mechanism $\mathcal M_3$ that with probability $p$ runs an ascending second-price auction on the grand bundle and with probability $(1-p)$ allocates a uniformly random agent a uniformly random item and runs an ascending second-price auction for the remaining item. 
\begin{claim}
    The mechanism $\mathcal{M}_3$ is universally OSP.  
\end{claim}
The proof is identical to the proof of Claim \ref{claim-mec-sm-mua-osp}. 
\begin{claim}
The mechanism $\mathcal{M}_3$ achieves a $\nicefrac 3 2$-approximation to the optimal social welfare.
\end{claim}
\begin{proof}
    Let $\{1,2\}$ denote the set of two bidders and $\{a,b\}$ denote the set of the two items. 
Throughout the proof, we slightly abuse notation by writing the value of a valuation $v$ for item $a$ as  $v(a)$ instead of $v_1(\{a\})$. 

The proof goes by a case  analysis.  First, we analyze the approximation guarantee of the mechanism if the optimal allocation allocates both items to the same agent, and then we analyze the case where the optimal allocation assigns a different item to each bidder. 

We begin with the first case:  assume that the optimal allocation assigns both $a$ and $b$ to the same bidder, say bidder $1$ without loss of generality. 
Note that if we run ascending auction on the grand bundle, then by assumption it outputs the optimal allocation.  Now, consider the case where $\mathcal M_3$ randomly allocates bidder $1$ with item $a$, and then we runs an ascending auction on item $b$.  Note that:
$v_1(a,b) \ge v_1(a)+v_2(b)$, so $v_1(a ~|~ b)\ge v_2(b)$. Therefore, the mechanism obtains welfare of $v_1(a)+\max\{v_1(b~|~a)+v_2(b)\}\ge v_1(a,b)=\opt$, i.e., $\mathcal M_3$ obtains the optimal welfare. Due to the same reasons, $\mathcal M_3$ obtains the optimal welfare also for the case where $b$ is randomly allocated to bidder $1$. 
Therefore, the expected welfare is  at least $\big(p+\frac{1-p}{2}\big)\cdot \opt$. 

Consider the complementary case where the optimal allocation assigns one item to each bidder, without loss of generality 
item $a$ to  $1$ and  $b$ to $2$. 
The expected welfare of $\mathcal M_6$ is:
\begin{align*}
    &p\cdot \underbrace{\max\{v_1(a,b),v_2(a,b)\}}_{\ge OPT/2} \quad&\text{(ascending auction on $\{a,b\}$)} &\\
    +&\frac{1-p}{4}\cdot\Big(v_1(a)+\max\big\{v_2(b),v_1(a,b)-v_1(a)\big\}\Big) &(1\gets a) \\ 
    +&\frac{1-p}{4}\cdot\Big(v_1(b)+
    \max\big\{v_2(a),v_1(a,b)-v_1(b)\big\} \Big) &(1\gets b) \\
    +&\frac{1-p}{4}\cdot\Big(v_2(a)+
    \max\big\{v_1(b),v_2(a,b)-v_2(a)\big\} \Big) &(2\gets a) \\
    +&\frac{1-p}{4}\cdot\Big(v_2(b)+
    \max\big\{v_1(a),v_2(a,b)-v_2(b)\big\} \Big) &(2\gets b)
\end{align*}
where $(1\gets a)$ signifies the case where bidder $1$ is randomly allocated with $a$ and an ascending auction is run on $b$, and $(1\gets b)$, $(2\gets a)$ and $(2\gets b)$ are defined analogously. 
Now, observe that since $\opt=v_1(a)+v_2(b)$, we have that: 
 $v_2(b)\ge v_1(b~|~a))$ and also $v_1(a)\ge v_2(a~|~b))$. Therefore, we obtain the optimal welfare in cases $(1\gets a)$ and $(2\gets b)$.  

 The contribution of the remaining cases,
where the bidders are allocated the \textquote{wrong} items, i.e., $(1\gets b)$ and $(2\gets a)$, to the expected welfare is: 
 \begin{multline*}
     \frac{1-p}{4}\cdot\Big(v_1(b)+
    \max\big\{v_2(a),v_1(a~|~b)\big\}+ v_2(a)+
    \max\big\{v_1(b),v_2(b~|~a)\big\}\Big)   \\  
    \ge \frac{1-p}{4}\cdot\Big( v_1(a,b)+v_2(a,b)  \Big) \ge \frac{1-p}{4}\cdot\big(v_1(a)+v_2(b)\big)=\frac{1-p}{4}\cdot \opt
 \end{multline*}
Therefore, the expected welfare is $\frac{p\cdot \opt}{2}$ + $\frac{3(1-p)\cdot \opt}{4}$.  
Since $p=\frac{1}{3}$, in both cases the expected welfare is at least $\frac{2}{3}$ of the optimum, which completes the proof. 
\end{proof}


\section{The Connection of Obviously Strategy-Proof Auctions and Weakly Group Strategy-Proof Mechanisms}\label{app:wgsp}
In his seminal work, \cite{li} has observed that  obviously strategy-proof mechanisms satisfy the quality of being weakly group strategy-proof\footnote{Loosely speaking, a mechanism $A$ is weakly group strategy-proof if no coalition of bidders can all simultaneously increase their utility by deviating from the dominant strategy.}, but the converse is not true: the top trading cycles mechanism is weakly group strategy-proof, but is not obviously strategy-proof (see also \cite{troyan2019obviously}). This gives rise to the question of understanding the differences between the two notions.

On one hand, weak group strategy-proofness is already
quite restrictive. For illustration, in certain settings (such as multi-unit auction with decreasing marginal valuations)  enforcing weak group strategy-proofness precludes exact welfare approximation \cite{GMR17}. 
This raises the question of how much more restrictive a notion can become beyond weak group strategy-proofness, and in particular whether obvious strategy-proofness is more restrictive than weak group strategy-proofness.  We ask:
are there settings in which OSP mechanisms must necessarily obtain worse approximation guarantees than mechanisms which are ``only'' weakly group strategy-proof? 

We answer this question affirmatively by examining the setting with heterogeneous items and unit-demand bidders.  First, as \cite{Ron24} has shown, deterministic obviously strategy-proof mechanisms that satisfy individual rationality and no negative transfers cannot obtain approximation better than $\min\{m,n\}$. Furthermore, as we show in \cref{thm-lb-ud},  one cannot achieve a $(\nicefrac 8 7-\varepsilon)$-approximation to the social welfare with any randomized OSP mechanism in this setting (similarly, assuming individual rationality and no negative transfers).  On the other hand, we will now explain why  for this setting the VCG mechanism (which is clearly optimal and deterministic) \emph{is} weakly group strategy-proof.  

This holds due to the combination of the following two well known facts.
First, the outcome corresponding to the minimum-price Walrasian equilibrium is weakly group strategy-proof in unit-demand settings even beyond quasi-linear utilities (see, e.g., \cite{morimoto2015strategy}). The second known fact is that for unit-demand settings the minimum-price Walrasian equilibrium corresponds to the VCG outcome and prices (see, e.g., Theorem 5 of \cite{gul1999walrasian}).
Therefore, we obtain a separation between the approximation ratio achievable by OSP mechanisms (which is $\min\{m,n\}$ for deterministic mechanisms \cite{Ron24}, and bounded away from $1$ for randomized ones) and the ratio of weakly group strategy-proof mechanisms (which is $1$ exactly).

We leave open the question of understanding the power of weakly group strategy-proof mechanisms for additional auction settings. In particular, we have yet to understand the power of weakly group strategy-proof mechanisms for additive bidders.  
For this setting, we cannot apply the arguments previously made for unit-demand bidders, as the the Vickrey-Clarke-Groves (VCG) mechanism 
is not weakly group strategy-proof\footnote{Take an example with two bidders and two items where the first bidder has value $10$ for both items and the second has value $9$ for both items.  By colluding and reporting that they only value distinct items, both bidders receive an item at a price of $0$ (which is preferable to the VCG outcome for both bidders).}.


\section{Missing Proofs from Section \ref{sec-prelims}}\label{app:missing-proofs-prems}

\begin{proof}[Proof of Lemma \ref{lemma-partial}]
    We establish the proof by describing an obviously dominant strategy for every bidder.
We argue that for each bidder $i$ the strategy which has the bidder remain in the auction if and only if her current clock price $p_i$ is less than or equal to $v_i(X_i^{P}) - v_i(X_i^{B})$, i.e., the increase in value she has for receiving $X_i^{P}$ compared to $X_i^{B}$, is obviously dominant. 
We prove it by a straightforward case analysis. In particular, we demonstrate that the worst case given the truthful strategy is always more profitable than the best case given any other strategy. 

Consider first the case where $p_i \leq v_i(X_i^{P}) - v_i(X_i^{B})$.  The best-case utility $i$ can receive by deviating from the strategy and exiting the auction early is $v_i(X_i^{B})$ (by receiving $X_i^{B}$ for a price of $0$), but the worst-case utility she receives by following the strategy and staying in the auction until the next price increment is also $v_i(X_i^{B})$ (since she can always exit the auction if the price then becomes too high).  

Now consider the case where $p_i > v_i(X_i^{P}) - v_i(X_i^{B})$.  The worst-case utility bidder $i$ can receive by following the strategy and exiting the auction is $v_i(X_i^{B})$. 
On the other hand, if she deviates from it by remaining an active player and the auction were to terminate she would obtain utility 
less than  $v_i(X_i^{P}) - (v_i(X_i^{P}) - v_i(X_i^B)) \leq v_i(X_i^B)$.  Thus, in this case as well, her worst case utility from following the strategy is better than the best case when deviating, which completes the proof.
\end{proof}

\begin{proof}[Proof of \cref{lem:yaos}] 
Fix a randomized mechanism $\mathcal M$ and 
denote with $\mathcal{D}_{\mathcal{M}}$ its distribution over the deterministic mechanisms in its support.
Note that by assumption every deterministic mechanism $A$ in the support of $\mathcal M$ satisfies that:
$$
    \E_{(v_1,\ldots,v_n)\sim {\mathcal{D}}}\Big[\dfrac{A(v_1,\ldots,v_n)}{OPT(v_1,\ldots,v_n)}\Big]\le 
    \frac{1}{\alpha}
$$
Since this is true for every deterministic mechanism in the support, averaging over the randomness of the mechanism $\mathcal M$:
$$
   \E_{(v_1,\ldots,v_n)\sim \mathcal{D},A\sim \mathcal {D}_M}\Big[\dfrac{A(v_1,\ldots,v_n)}{OPT(v_1,\ldots,v_n)}\Big]\le \frac{1}{\alpha} 
$$
Denote with $(v_1^\ast,\ldots,v_n^\ast)$ the valuation profile in the support of $\mathcal D$ that minimizes $\E_{A\sim \mathcal \mathcal D_M}\Big[\frac{A(v_1^\ast,\ldots,v_n^\ast)}{OPT(v_1^\ast,\ldots,v_n^\ast)}\Big]$. Since the minimum is  smaller or equal than the average:
$$
   \E_{A\sim \mathcal \mathcal D_M}\Big[\dfrac{A(v_1^\ast,\ldots,v_n^\ast)}{OPT(v_1^\ast,\ldots,v_n^\ast)}\Big]\le  \frac{1}{\alpha}
$$
Thus, the randomized mechanism $\mathcal M$ does not give an  approximation better than $\alpha$ to the optimal welfare given the valuation profile $(v_1^\ast,\ldots,v_n^\ast)$, which completes the proof.
\end{proof}

\section{Missing Proofs from Section \ref{sec-mua}: Multi-Unit Auctions}\label{app-missing-mua}
\subsection{Proof of Theorem \ref{claim:ascending-bundles-approx}}
Fix a valuation profile $(v_1,d_1),\ldots,(v_n,d_n)$ where $d_i$ is the desired set and $v_i$ is the value for them and
consider a social welfare optimizing allocation $\vec O=(O_1,\ldots,O_n)$, whose  welfare we denote with $\opt$. In particular, Let $\vec O$ be the allocation that allocates the minimum number of items among the optimal ones.\footnote{Note that we slightly abuse notation here, since $v_i$ also stands for a valuation function.}   

Consider partitioning the bidders who receive at least one item in this allocation into $k+1$ disjoint subsets based on the number of items allocated to them 
such that bidders in group $i\in \{0,1, 2, \dots, k=\lceil \log{m} \rceil-1\}$ receive between $2^{i}+1$ and $2^{i+1}$ items. We denote the subset of bidders in group $i$ with $\mathcal B_i$. 

Denote with $\mathcal U[0,k]$ the uniform distribution over the integers $\{0,\ldots,k\}$. Observe that the fact there exist at most $\log m$ groups implies that: 
\begin{equation}\label{eq-opt-lb}
\mathbb{E}_{i \sim \mathcal{U}[0, k]} \big[\sum_{i \in \mathcal B_i} v_i\big] = 
\frac{1}{\log m} \cdot \sum_i v_i(O_i)= \frac{OPT}{\log m}
\end{equation}

For every $i\in \{0,\ldots,k\}$, denote with $\delta_i$ the total number of bidders whose demand $d_i$ is smaller than or equal to $2^{i+1}$ (including bidders who are not allocated in the optimal allocation).  Observe that since
each bidder in the group $\mathcal B_i$ receives between $2^{i}+1$ and $2^{i+1}$ items in the optimal allocation, the number of the bidders in each subset  $\mathcal B_i$ is at most 
$\min\{\Big\lfloor \frac{m}{2^{i}+1}\Big\rfloor ,\delta_i\}$.  

Now, fix an integer $i\in \{0,1,\ldots,k\}$ and
consider the allocation of \textsc{Random-Bundles} whenever it allocates bundles of size $2^{i+1}$. Denote the bidders who are satiated in this allocation with $\mathcal A_i$. We make the following two observations regarding this set of bidders.

First, note that by construction the bidders in $\mathcal A_i$ are the bidders with the highest values who are  satiated by bundles of size  $2^{i+1}$. Since the bidders in $\mathcal B_i$ are also satiated by bundles of size $2^{i+1}$, we get that every bidder who is in $\mathcal A_i$ has a higher value than any bidder who is in $\mathcal B_i\setminus \mathcal A_i$. Formally: 
\begin{equation}\label{eq-diff}
  \max_{i\in \mathcal B_i \setminus \mathcal A_i} v_i \le \min_{i\in \mathcal A_i} v_i  
\end{equation}
Second, since \textsc{Random-Bundles} allocates bundles of size $2^{i+1}$, it implies that $|\mathcal A_i|\ge \min\{\lfloor \frac{m}{2^{i+1}}\rfloor,\delta_i\}$ and therefore: 
\begin{equation}\label{eq-diff-count}
2|\mathcal A_i|\ge \min \{\Big\lfloor \frac{m}{2^{i}}\Big\rfloor,2\cdot \delta_i\} \ge |\mathcal B_i|
\end{equation}
where we remind that the rightmost inequality holds because $|\mathcal B_i|\le \min\{\Big\lfloor \frac{m}{2^{i}+1}\Big\rfloor ,\delta_i\}$.
Combining (\ref{eq-diff}) and (\ref{eq-diff-count}) gives that for every subset of bidders $\mathcal B_i$:
\begin{equation}\label{eq-comp}
    \sum_{i\in \mathcal B_i} v_i \le    \sum_{i\in \mathcal B_i\cap \mathcal A_i} v_i  + 
        \sum_{i\in \mathcal B_i\setminus \mathcal A_i} v_i \le \sum_{i\in \mathcal A_i} v_i + |\mathcal B_i \setminus \mathcal A_i|\cdot \min_{i\in \mathcal A_i} v_i \le 3\cdot \sum_{i\in \mathcal A_i} v_i
\end{equation}
where the rightmost inequality holds since clearly $|\mathcal B_i\setminus \mathcal A_i| \le |\mathcal B_i|$.

Now, denote with $ALG$ the expected welfare of 
\textsc{Random-Bundles}. We remind that \textsc{Random-Bundles} samples $i\in \{0,\ldots,k\}$ uniformly at random, and therefore:
\begin{align*}
     ALG &=  \mathbb{E}_{i \sim \mathcal{U}[0, k]} \Big[\sum_{i \in \mathcal A_i} v_i\Big] \\
& \ge \mathbb{E}_{i \sim \mathcal{U}[0, k]} \Big[\frac{\sum_{i \in \mathcal B_i} v_i}{3}\Big] &\text{(by (\ref{eq-comp}))} \\
 &= \frac{OPT}{3\log m} &\text{(by (\ref{eq-opt-lb}))} 
\end{align*}
which completes the proof.

\subsection{Proof of Lemma \ref{lem:auction-sampling}}\label{subsec::proof-auction-sampling}

To prove \cref{lem:auction-sampling}, we rely on the following useful and general lemma from \cite{bei2017worst}.  
\begin{lemma}[Lemma 2.1 from \cite{bei2017worst}]
Consider any subadditive function $f : A \rightarrow \mathbb{R}$.  For a given subset $S \subseteq A$ and a positive integer $k$ we assume that $f(S) \geq k \cdot f(\{i\})$ for any $i \in S$.  Further, suppose that $S$ is divided uniformly at random into two groups $T_1$ and $T_2$.  Then, with probability of at least $1/2$, we have $f(T_1) \geq \frac{k-1}{4k} \cdot  f(S)$ and $f(T_2) \geq \frac{k-1}{4k} \cdot f(S)$.
\end{lemma}

Applying this lemma to our auction setting, we can prove our desired statement.

\begin{proof}[Proof of Lemma \ref{lem:auction-sampling}]
To show the desired statement we first argue that the optimal welfare $f$ achievable in any  auction setting (either with homogeneous or heterogeneous items) is a \emph{subadditive} function \emph{over the bidders}. To this end,  consider a set of bidders $S$ and the optimal solution over these bidders.  Observe that each item is allocated to any bidder at most once.  Thus, for any two sets $T_1$ and $T_2$ such that $T_1 \cup T_2 = S$ we have that any feasible solution $x$ in $S$ comprises bidders which appear either in $T_1$ or $T_2$ (or possibly both).  Since we can then construct feasible solutions in $T_1$ and $T_2$ which capture the allocation $x$ it must be that the value of the optimal solution in $T_1$ plus the value of the optimal solution in $T_2$ is greater than or equal to the value of the optimal solution in $S$.

But then, since we have that the optimal welfare on an instance of combinatorial auctions is subadditive over the ground set of bidders and we assume that no bidder is critical, we have that the $k$ in the statement of Lemma 2.1 of \cite{bei2017worst} is $100$.  As such, we have that when we randomly partition the bidders into unsampled and sampled sets, both sets have an optimal welfare that is sets are within a factor $99/400 > 1/5$ of the optimal welfare with probability at least $1/2$.
\end{proof}

\subsection{Proof of Lemma \ref{lemma-small-pay}}\label{subsec::proof-lemma-small-pay}
The proof is a direct consequence of the approximation guarantees of the mechanisms considered, as well 
as  the fact that they are
obviously strategy-proof (and thus also dominant-strategy incentive compatible) and satisfy individual rationality and no negative transfers. 
The proof is very similar to a proof in \cite{Ron24} and we write it for completeness.

\begin{proof}[Proof of Lemma \ref{lemma-small-pay}]
For part \ref{item-1}, note that because of individual rationality:
\begin{equation}\label{eq-1}
 v_1^{one}(f(v_1^{one},v_2^{one}))-P_1(v_1^{one},v_2^{one})\ge 0   
\end{equation}
We remind that by Claim \ref{claim-mua-sm-instances} part \ref{condi-1}, the allocation rule $f$ allocates to bidder $1$ at least one item given $(v_1^{one},v_2^{one})$, so:
\begin{equation}\label{eq-2}
 v_1^{one}(f(v_1^{one},v_2^{one}))=1   
\end{equation}
Combining (\ref{eq-1}) and (\ref{eq-2}) gives part \ref{item-1}.

For part \ref{item-2}, note that given $(v_1^{ONE},v_2^{ALL})$, Claim \ref{claim-mua-sm-instances} part \ref{condi-2} implies that
player $1$ gets the empty bundle. Because of individual rationality: 
$$v_1^{ONE}(f(v_1^{ONE},v_2^{all}))-P_1(v_1^{ONE},v_2^{all})\ge 0 $$
Since $v_1^{ONE}(f(v_1^{ONE},v_2^{all}))=0$, we get that $0\ge P_1(v_1^{ONE},v_2^{all})$, and because of the no negative transfers property, we get $P_1(v_1^{ONE},v_2^{all})=0$, as needed.

For part \ref{item-3}, note that by Claim \ref{claim-mua-sm-instances} part \ref{condi-3}, player $1$ wins all of the items given $(v_1^{all},v_2^{one})$, so $v_1^{all}(f(v_1^{all},v_2^{one}))=k^2$. Combining this fact with individual rationality implies that $P_1(v_1^{all},v_2^{one})\le k^2$. Therefore:
\begin{equation*}
    v_1^{ALL}(f(v_1^{all},v_2^{one}))-P_1(v_1^{all},v_2^{one})\ge k^4 - k^2
\end{equation*}
Since the mechanism is obviously strategy-proof, it is therefore also dominant-strategy incentive compatible, we get that: $    v_1^{ALL}(f(v_1^{ALL},v_2^{one}))-P_1(v_1^{ALL},v_2^{one})\ge k^4 - k^2
$. The only valuable bundle for player $1$ given $v_1^{ALL}$ is the grand bundle, so we get that 
$f(v_1^{ALL},v_2^{one})$ allocates all items to player $1$. In addition, it implies that the payment of player $1$ given $(v_1^{ALL},v_2^{one})$ is necessarily at most $k^2$, which completes the proof. \qedhere

\end{proof}

\subsection{Proof of Theorem \ref{thm:decreasing-marginals}}
\label{subsec::proof--dec-mua}
\begin{proof}
First of all, the mechanism is obviously strategy-proof, and the proof of it is identical to the proof of \cref{claim-mua-sm-mechanism-osp}. 

We now proceed to prove the approximation guarantee of the mechanism, following an approach that closely resembles the proof of \cref{lem:single-minded-approx}. We remind that a bidder $i$ is \emph{critical} if allocating to her the grand bundle
gives a $1/100$-approximation to the optimal welfare. 
Thus, if there exists a critical bidder, we obtain a $1/200$-approximation by running an ascending auction for the grand bundle with probability at least $1/2$.

    In the case that there does not exist a critical bidder, we may again appeal to our sampling lemma, i.e., Lemma \ref{lem:auction-sampling}, to show that when we sample bidders, we obtain an ``accurate enough'' estimates with probability at least $1/2$. Formally, with probability $1/2$ we have that the social welfare $\text{OPT}(S)$ contained in the sampled set $S$ is between $\text{OPT}/5$ and $\text{OPT}$  and the same holds for $\text{OPT}(U)$, the welfare among the unsampled bidders. 
    Therefore, with probability $\frac 1 2 $ we set a per-item price $p$ which is in the range $[\text{OPT}/(50m), \text{OPT}/(10m)]$. We will now further refine the case analysis by examining the number of items sold.
    
    The \textquote{easy} case is if we sell at least $\nicefrac m 2$ items. 
      Since an unsampled bidder buying $t$ goods spends at least $\frac{t\text{OPT}}{50m}$, their value for the purchased bundle is at least $\frac{t\text{OPT}}{50m}$. Therefore, the total value of all bidders who purchase goods is at least $\frac{m}{2}\cdot\frac{\text{OPT}}{50m} = \frac{\text{OPT}}{100}$. Therefore,  we  obtain welfare of at least $\text{OPT}/100$. Altogether, since we run uniform sampling with probability $\nicefrac{1}{2}$ and the estimation is
\textquote{good} with probability $\nicefrac{1}{2}$, we obtain  a $400$-approximation to the welfare.

    Suppose, by contrast, that we sell fewer than $\nicefrac m 2$ items to the unsampled bidders.  
    To analyze this case, let $\vec{q}=(q_1,\dots,q_n)$ be the optimal allocation if the items are divided only among the bidders in $U$ (clearly, every bidder not in $U$ is allocated zero items, and the welfare of $\vec q$ is equal to $\text{OPT}(U)$). 
    Observe that the welfare of $\vec q$
is partitioned to ``low'' marginal values, i.e., marginal values less than or equal to $p$ and ``high'' marginal values, which are greater than $p$.

Observe that since Mechanism \ref{alg:single-minded} allocates less than $\frac m 2$ items, every bidder in  $U$  could have bought additional items at a price of $p$. Therefore, all bidders who were allocated in $\vec q$ were also allocated by Mechanism \ref{alg:single-minded} all the items for which they had \textquote{high} marginal values. Therefore, it remains to bound the loss coming from \textquote{low} margins. Observe that the marginal value each agent in $U$ has for receiving an additional good is no more than $\text{OPT}/(10m)$ and since there are $m$ items in total allocated in $\vec q$, the total welfare of $\vec q$ coming from \textquote{low} marginals is at most $\text{OPT}/10$. However, by assumption the total welfare of $\vec q$ is at least $\text{OPT}/5$,  so at least 
$\text{OPT}/10$ of welfare comes from \textquote{high} marginals which also contribute to the welfare of the allocation of Mechanism \ref{alg:single-minded}. As we said before, this 
depends on finding a good partition of $U,S$ and running a uniform price auction which occurs in probability $\nicefrac{1}{4}$, 
so overall the expected welfare of the mechanism in this case is at least $\frac{OPT}{40}$.

Combining all cases, we conclude that the expected welfare of Mechanism \ref{alg:single-minded} is at least $\frac{OPT}{400}$,  thereby completing the proof.
\end{proof}


\subsection{ Proof of  Theorem \ref{thm-lb-mua-dec}: Impossibility for 2 Items and 2 Bidders} \label{subsec-lb-proof-mua-dec}
First of all, We assume that the domain $V_i$ of each bidder consists of valuations with values in  $\{0,1,\ldots,k^4\}$ that satisfy decreasing marginal utilities, where $k$ is an arbitrarily large number. 

Assume towards a contradiction that 
there exists an obviously strategy-proof, individually rational, no negative transfers mechanism $A$ together with strategy profile $\mathcal S=(\mathcal S_1,\mathcal S_2)$
that implement an allocation rule and payment schemes  $(f,P_1,P_2):V_1\times V_2\to \allocs\times \mathbb R^2$, where $f$ 
gives an approximation strictly better than $2$ to the optimal social welfare.  
For every player $i$, we define three valuations that will be of particular interest:  
$$
v_i^{ALL}(x)=\begin{cases}
k^2 &\quad x=1, \\
2k^2 &\quad x=2.
\end{cases} \\\quad 
v_i^{one}(x)= \begin{cases}
1 \quad x=1, \\
1 \quad x=2.
\end{cases}
 \quad
v_i^{ONE}(x)=\begin{cases}
4k \quad x= 1, \\
4k  \quad x=2.
\end{cases} 
$$
where $k$ is arbitrarily large.

For the analysis of the mechanism, we define the following subsets of valuations:
$
\mathcal{V}_1=\{v_1^{one},v_1^{ONE},\allowbreak v_1^{ALL}\}\allowbreak \subseteq V_1$ and  $\mathcal{V}_2=\{v_2^{one},v_2^{ONE},v_2^{ALL}\}\subseteq V_2$.  
Observe that:
\begin{claim}\label{claim-mua-dec-instances}
Every deterministic mechanism that has approximation strictly better than $2$ necessarily satisfies the following  conditions simultaneously:
\begin{enumerate}
    \item Given the valuation profile $(v_1^{one},v_2^{one})$, the mechanism
    allocates one item to bidder $1$ and one item to bidder $2$. \label{condi-1-dec}
    \item Given the valuation profile $(v_1^{ONE}, v_2^{ALL})$, bidder $2$ wins all items.
    \label{condi-3-dec}
\end{enumerate}
\end{claim}
The proof of \cref{claim-mua-dec-instances} is straightforward: if a deterministic mechanism does not satisfy one of conditions, then due the fact that $k$ is arbitrarily large implies that the approximation guarantee of the mechanism is at most $2$ in the worst case.

We now focus on 
$
\mathcal{V}_1\times \mathcal{V}_2$. 
Observe that there necessarily exists a vertex $u$, and valuations $v_1,v_1' \in \mathcal{V}_1$, and  $v_2,v_2' \in \mathcal{V}_2$ such that $(\mathcal{S}_1(v_1), \mathcal{S}_2(v_2))$ diverge at vertex $u$. This follows from \cref{claim-mua-dec-instances}, which implies that the mechanism $A$ must output different allocations for different valuation profiles in $\mathcal{V}_1 \times \mathcal{V}_2$. Consequently, not all valuation profiles end up in the same leaf, meaning that divergence must occur at some point. 


Let $u$ be the first vertex in the protocol such that 
 $(\mathcal{S}_1(v_1),\mathcal{S}_2(v_2))$ and $(\mathcal{S}_1(v_1'),\mathcal{S}_2(v_2'))$ diverge, i.e., dictate different messages. 
Note that by definition this implies that $u\in Path(\mathcal{S}_1(v_1),\mathcal{S}_2(v_2))\cap Path(\mathcal{S}_1(v_1'),\mathcal{S}_2(v_2'))$ and that either bidder $1$ or bidder $2$ sends different messages for the valuations in $\mathcal{V}_1$ or $\mathcal V_2$, respectively. 
Without loss of generality, we assume that bidder $1$ sends different messages, meaning that there exist $v_1,v_1'\in \mathcal{V}_1$ such that $\mathcal S_1(v_1)$ and $\mathcal S_1(v_1')$ dictate different messages at vertex $u$.  
We remind that $\mathcal{V}_1=\{v_1^{one},v_1^{ONE},v_1^{all}\}$, so the  following claims jointly imply a contradiction, completing the proof: 

\begin{claim}\label{claim-oneone-same-dec}
    The strategy $\mathcal S_1$ dictates the same message at vertex $u$ for the valuations $v_1^{one}$ and $v_1^{ONE}$. 
\end{claim}
\begin{claim}\label{claim-one-all-same-dec}
        The strategy $\mathcal S_1$ dictates the same message at vertex $u$ for the valuations $v_1^{ONE}$ and $v_1^{ALL}$.
\end{claim}
In the proofs of Claims \ref{claim-oneone-same-dec} and Claim \ref{claim-one-all-same-dec}, we use the following lemma:
\begin{lemma}\label{lemma-small-pay-dec}
    The allocation rule $f$ and the payment scheme $P_1$ of bidder $1$ satisfy that:
    \begin{enumerate}
        \item Given $(v_1^{one},v_{2}^{one})$ bidder $1$ wins  one item and pays at most $1$.  \label{item-1-dec}
        \item  Given $(v_1^{ONE},v_2^{ALL})$, bidder $1$ gets the empty bundle and pays zero.   \label{item-2-dec}
        \item Given $(v_1^{ALL},v_2^{one})$, bidder wins all items and pays at most $2k$. \label{item-3-dec}  
    \end{enumerate}
\end{lemma}
The lemma is a direct consequence of the properties of the mechanism. 
We use 
it to prove \cref{claim-oneone-same-dec,claim-one-all-same-dec}
and defer the proof to  \cref{sec-small-pay-proof-alltogether-dec}.  The proofs are identical to those in \cite{Ron24}, and we include them here for completeness.
\begin{proof}[Proof of Claim \ref{claim-oneone-same-dec}]
    Note that by Lemma \ref{lemma-small-pay-dec} part \ref{item-1-dec}, $f(v_1^{one},v_2^{one})$ allocates at least one item to player $1$ and $P_1(v_1^{one},v_2^{one})\le 1$. Therefore:
\begin{equation}\label{eq-good-leaf1-dec}
 v_1^{ONE}(f(v_1^{one},v_2^{one}))-P_1(v_1^{one},v_2^{one})\ge 4k-1 
\end{equation}
 In contrast, by part \ref{item-2-dec} of Lemma \ref{lemma-small-pay-dec},   $f(v_1^{ONE},v_2^{ALL})$ allocates no items to player $1$ and $P_1(v_1^{ONE},v_2^{ALL})=0$, so:
 \begin{equation}\label{eq-bad-leaf1-dec}
 v_1^{ONE}(f(v_1^{ONE},v_2^{ALL}))-P_1(v_1^{ONE},v_2^{ALL})= 0   
\end{equation}
Combining inequalities (\ref{eq-good-leaf1-dec}) and (\ref{eq-bad-leaf1-dec}) gives:
\begin{equation*}
  v_1^{ONE}(f(v_1^{ONE},v_2^{ALL}))-P_1(v_1^{ONE},v_2^{ALL})< 
  v_1^{ONE}(f(v_1^{one},v_2^{one}))-P_1(v_1^{one},v_2^{one})  
\end{equation*}
We remind that vertex $u$ belongs in $Path(\mathcal S_1(v_1^{one}),\mathcal S_2(v_2^{one}))$ and also in
$Path(\mathcal{S}_1(v_1^{ONE}),\allowbreak\mathcal{S}_2(v_2^{ALL}))$. Therefore, Lemma \ref{lemma-bad-leaf-good-leaf} gives that the strategy $\mathcal S_1$ dictates the same message for  $v_1^{one}$ and $v_1^{ONE}$ at vertex $u$.
\end{proof}

\begin{proof}[Proof of \cref{claim-one-all-same-dec}]
    Following the same approach as in the proof of Claim \ref{claim-oneone-same-dec}, note that by \cref{lemma-small-pay-dec}  \cref{item-3-dec}: 
\begin{equation}\label{break-align-dec}
v_1^{ONE}(f(v_1^{ALL},v_2^{one}))-P_1(v_1^{ALL},v_2^{one}) \ge 4k-2k=2k     
\end{equation}
Also, by \cref{lemma-small-pay-dec}  \cref{item-2}:
\begin{equation}\label{break-align2-dec}
  v_1^{ONE}(f(v_1^{ONE},v_2^{ALL}))  
-P_1(v_1^{ONE},v_2^{ALL})=0   
\end{equation}
Combining \cref{break-align-dec} and \cref{break-align2-dec}:
\begin{equation*}
    v_1^{ONE}(f(v_1^{ONE},v_2^{ALL}))  
-P_1(v_1^{ONE},v_2^{ALL})< v_1^{ONE}(f(v_1^{ALL},v_2^{one}))-P_1(v_1^{ALL},v_2^{one})
\end{equation*}
Combining the above inequality with the fact that 
vertex $u$ belongs in $Path(\mathcal S_1(v_1^{ALL}),\mathcal S_2(v_2^{one}))$ and in
$Path(\mathcal{S}_1(v_1^{ONE}),\mathcal{S}_2(v_2^{ALL}))$ implies that by Lemma \ref{lemma-bad-leaf-good-leaf} 
the obviously dominant strategy $\mathcal S_1$ dictates the same message for the valuations $v_1^{ONE}$ and $v_1^{ALL}$ at vertex $u$. 
\end{proof}

\subsection{Proof of Lemma \ref{lemma-small-pay-dec}: Observations About The Mechanism}\label{sec-small-pay-proof-alltogether-dec}
The proof is a direct consequence of the approximation guarantee of the mechanism and the fact that it is 
obviously strategy-proof (and thus also dominant-strategy incentive compatible) and satisfies individual rationality and no negative transfers. 



For part \ref{item-1}, note that because of individual rationality:
\begin{equation}\label{eq-1-dec}
 v_1^{one}(f(v_1^{one},v_2^{one}))-P_1(v_1^{one},v_2^{one})\ge 0   
\end{equation}
We remind that by \cref{claim-mua-dec-instances} \cref{condi-1-dec}, the allocation rule $f$ allocates to bidder $1$ one item given $(v_1^{one},v_2^{one})$, so:
\begin{equation}\label{eq-2-dec}
 v_1^{one}(f(v_1^{one},v_2^{one}))=1   
\end{equation}
Combining (\ref{eq-1-dec}) and (\ref{eq-2-dec}) gives part \ref{item-1-dec}.

For part \ref{item-2-dec}, note that 
by \cref{claim-mua-dec-instances} \cref{condi-3-dec}, 
given $(v_1^{ONE},v_2^{ALL})$ player $1$ gets the empty bundle. Because of individual rationality: 
$$v_1^{ONE}(f(v_1^{ONE},v_2^{ALL}))-P_1(v_1^{ONE},v_2^{ALL})\ge 0 $$
Since $v_1^{ONE}(f(v_1^{ONE},v_2^{ALL}))=0$, we get that $0\ge P_1(v_1^{ONE},v_2^{ALL})$, and because of the no negative transfers property, $P_1(v_1^{ONE},v_2^{ALL})=0$, as needed.  

To prove part \ref{item-3-dec}, we define another valuation:
$$
\hat{v}_1(x)=\begin{cases}
    k \quad &x=1, \\
    2k \quad &x=2.
\end{cases}
$$
Given $(\hat{v}_1,v_2^{one})$, player $1$ wins $2$ items because  
$f$ gives an approximation strictly better than $2$. Thus, the inequality $\hat{v}_1(f(\hat{v}_1,v_2^{one}))-P_1(\hat{v}_1,v_2^{one})\ge 0$ holds because of individual rationality, and therefore $P_1(\hat{v}_1,v_2^{one})\le 2k$. 

Note that $f$ clearly also allocates all items to player $1$ given $({v_1}^{ALL},v_2^{one})$ because of its approximation guarantee. The fact that the mechanism is dominant-strategy incentive compatible and $f$ outputs the same allocation for both $(\hat{v}_1,v_2^{one})$ and  $({v_1}^{ALL},v_2^{one})$ implies that $P_1({v_1}^{ALL},v_2^{one})=P_1(\hat{v}_1,v_2^{one})$, so $P_1({v_1}^{ALL},v_2^{one})$ is also smaller than $2k$, which completes the proof. 

\subsection{Proof of Lemma \ref{lemma:mono-mua-dec}: A  Deterministic Mechanism For 3 Items and 2 Bidders}\label{sec-impos-mua-dec}
    Consider the following mechanism: allocate one item to each bidder, and run an ascending auction on the remaining item. Assume tie breaking is in favor of player $1$, i.e., if both players have the same value for the item, then player $1$ wins it.
    This mechanism is a generalized ascending auction, so by \cref{lemma-partial} it is obviously strategy-proof. 
    
    We now analyze its approximation guarantee. 
Intuitively, the mechanism guarantees a $1.5$ approximation because the worst  case scenario is that some bidder  should  have gotten all three items in the optimal solution, but gets instead only two items. The decreasing marginal property ensures that the loss is bounded by at most $\frac{1}{3}$ of the optimal social welfare.



 
    
    Formally, 
denote with $(q_1,q_2)$ the allocation of the mechanism, and let $(o_1,o_2)$ be a welfare-maximizing allocation. We use the following notations: $$ALG=v_1(q_1)+v_2(q_2),\quad OPT=v_1(o_1)+v_2(o_2)$$

    Note that the algorithm necessarily outputs an allocation where one bidder gets $2$ items and the other bidder gets $1$ item. Assume without loss of generality that the ascending auction gives bidder $1$ two items and bidder $2$ one item, meaning that  $q_1=2$ and $q_2=1$. Note that by the definition of the auction:
    \begin{equation}\label{eq-opt-ge-alg}
        v_1(2)-v_1(1)\ge v_2(2)-v_2(1) 
    \end{equation}
To show that $ALG \ge \frac{2}{3}\cdot OPT$,  we proceed with the following case analysis:
\paragraph{Case I: $o_1=3$ and $o_2=0$.}
We will first show that $v_1(3)-v_1(2)\le \frac{OPT}{3}$, and then show why it implies that $ALG \ge \frac{2}{3}\cdot OPT$.
Observe that:
\begin{equation*}\label{eq-case-30}
    OPT=v_1(3)=v_1(3)-v_1(2)+v_1(2)-v_1(1)+v_1(1) \ge 3\cdot \big(v_1(3)-v_1(2)\big)
\end{equation*}
where the inequality holds because $v_1$ has decreasing marginal values. 
Therefore:
$$
OPT=v_1(3)=v_1(2)+v_1(3)-v_1(2)\le ALG +v_1(3)-v_1(2)\le ALG + \frac{OPT}{3}
$$
So $ALG \ge \frac{2\cdot OPT}{3}$, as needed.


\paragraph{Case II: $o_1=2$ and $o_2=1$.} In this case, the optimal allocation is identical to the allocation of the mechanism, so $ALG=OPT$ holds trivially .

\paragraph{Case III: $o_1=1$ and $o_2=2$.} Observe that $ALG=v_1(2)+v_2(1)\ge v_1(1)+v_2(2)=OPT$, where the inequality is by \cref{eq-opt-ge-alg}.


\paragraph{Case IV: $o_1=0$ and $o_2=3$.} Due to the same explanation as in case III, we have that: 
\begin{equation}\label{case-03}
    ALG=v_1(2)+v_2(1)\ge v_1(1)+v_2(2)
\end{equation}
And also that:
\begin{align*}
v_1(1) &\ge v_1(2)-v_1(1) &\text{(due to decreasing margins)} \\
&\ge v_2(2)-v_2(1) &\text{(by \cref{eq-opt-ge-alg})} \\
&\ge v_2(3)-v_2(2) &\text{(due to decreasing margins)} \\ 
\end{align*}
So $v_1(1)+v_2(2)\ge v_2(3)=OPT$. Combining this with \cref{case-03} gives that $ALG \ge OPT$, as needed.  

\section{Missing Proofs from Section \ref{sec-combi}: Combinatorial Auctions} \label{app-missing-combi}
\subsection{Proof of Lemma \ref{lemma-add-osp}
} \label{subsec::proof-add-osp}
Fix any realization of the coin flips of the mechanism.  Observe that ``sampled'' bidders receive $0$ utility regardless of their report so
reporting their valuations trtuhfully is obviously dominant for them. 
Further the ``unsampled'' bidders purchase all available items for which they have  positive utility and purchase no items for which they have negative utility when being truthful, so for them as well truthfulness is an obviously dominant strategy. 

\subsection{Proof of Theorem \ref{thm-lb-ud}: An Impossibility for Unit-Demand Bidders} \label{lb-ud-proof-place}

The proof follows the same structure as in \cref{thm-mua-sm-lb}: roughly speaking, we describe a distribution $\mathcal{D}$, show that it is “hard” for deterministic mechanisms, and then use Yao's Lemma to deduce hardness for randomized mechanisms.

The proof outline is as follows. In \cref{subsubsec-description-ud}, 
we describe the  hard distribution $\mathcal D$. 
In \cref{subsubsec-performance-ud} we analyze the allocation and payments of a deterministic \textquote{good} mechanism 
and explain why these properties  imply that  a \textquote{good} mechanism does not exist in  \cref{subsubsec-contradiction-ud}.
We defer technical proofs to \cref{app-claims-proofs-ud,subsubsec-prop-alloc-ud}. 

We prove for the case of two bidders and two items, but the proof extends to any number of bidders and any number of items by adding bidders with the all-zero valuation and assuming that the bidders in our construction have zero values for the additional items.

\subsubsection[Construction of a ``Hard'' Distribution D]{Construction of a ``Hard'' Distribution $\mathcal{D}$}\label{subsubsec-description-ud}
    To define the probability distribution over the valuation profiles, we name the two items $a$ and $b$, and let $k$ be an arbitrarily large number. Note that since these valuations are unit-demand, we can fully describe them by specifying their value per item.
    Consider the following valuations:  
\[
v_1^{{a,one}}(x) = 
\begin{cases}
1 & x=a,\\
0 & \text{otherwise.}
\end{cases} \quad 
v_2^{{b,one}}(x) = 
\begin{cases}
1 & x=b,\\
0 & \text{otherwise.}
\end{cases}
\]

\[
v_1^{{a,mid}}(x) = 
\begin{cases}
k^2 & x=a,\\
0 & \text{otherwise.}
\end{cases} \quad 
v_2^{{b,mid}}(x) = 
\begin{cases}
k^2 & x=b,\\
0 & \text{otherwise.}
\end{cases}
\]

\[
v_1^{{a,large}}(x) = 
\begin{cases}
k^4 & x=a,\\
0 & \text{otherwise.}
\end{cases} \quad 
v_2^{{b,large}}(x) = 
\begin{cases}
k^4 & x=b,\\
0 & \text{otherwise.}
\end{cases}
\]

\[
v_1^{{b,small}}(x) = 
\begin{cases}
k & x=b,\\
0 & \text{otherwise.}
\end{cases} \quad 
v_2^{{a,small}}(x) = 
\begin{cases}
k & x=a,\\
0 & \text{otherwise.}
\end{cases}
\]

\[
v_1^{{b,mid}}(x) = 
\begin{cases}
k^2 & x=b,\\
0 & \text{otherwise.}
\end{cases} \quad 
v_2^{{a,mid}}(x) = 
\begin{cases}
k^2 & x=a,\\
0 & \text{otherwise.}
\end{cases}
\]

\[
v_1^{{b,large}}(x) = 
\begin{cases}
k^4 & x=b,\\
0 & \text{otherwise.}
\end{cases} \quad 
v_2^{{a,large}}(x) = 
\begin{cases}
k^4 & x=a,\\
0 & \text{otherwise.}
\end{cases}
\]

\[
v_1^{{both}}(x) = 
\begin{cases}
2k+3 & x=a,\\
2k+1 & x=b,\\
0 & \text{otherwise.}
\end{cases} \quad 
v_2^{{both}}(x) = 
\begin{cases}
2k+3 & x=b,\\
2k+1 & x=a,\\
0 & \text{otherwise.}
\end{cases}
\]

Consider the following valuation profiles: 
\begin{align*}
 I_1 = (v_1^{{a,one}}&, v_2^{{b,one}}),\quad  I_2 = (v_1^{a,mid}, v_2^{{a,large}}),  \quad I_3 = (v_1^{b,large}, v_2^{b,mid}), \quad 
 I_4 = (v_1^{b,mid}, v_2^{b,large}) \\ &I_5 = (v_1^{a,large}, v_2^{a,mid}) \quad I_6 = (v_1^{b,small}, v_2^{b,one}) \quad I_7 =(v_1^{a,one}, v_2^{a,small}) 
\end{align*}
Let $\mathcal D$ be the distribution over valuation profiles where the probability of the valuation profile $I_1$ is $\frac{1}{4}$, and the probability of all the valuation profiles $I_2,I_3,I_4,I_5,I_6$ and $I_7$ is $\frac{1}{8}$ each. 

\subsubsection[The Performance of the Deterministic Mechanism A on the "Hard" Distribution D]{The Performance of the Deterministic Mechanism $A$ on the ``Hard'' Distribution $\mathcal{D}$}\label{subsubsec-performance-ud}
We remind that our goal is to show that no deterministic mechanism that satisfies all of the desired properties extracts more than $\frac{7}{8}$ of the optimal welfare. For that, we begin by observing that:
\begin{lemma}\label{lemma-ud-instances}
Every deterministic mechanism that has approximation better than $\frac{7}{8}$ necessarily satisfies all of the following conditions:
\begin{enumerate}
    \item Given the valuation profile $I_1 = (v_1^{\text{a,one}}, v_2^{\text{b,one}})$, the mechanism
    allocates item $a$ to player $1$ and item $b$ to player $2$. \label{condi-1-ud}
\item Given the valuation profile $I_2 = (v_1^{a,mid}, v_2^{a,large})$, the mechanism allocates item $a$ to bidder $2$. \label{condi-2-ud}

    \item Given the valuation profile $I_3 = (v_1^{b,large}, v_2^{b,mid})$, the mechanism allocates item $b$  to bidder $1$.
    \label{condi-3-ud}
\item Given the valuation profile $I_4 = (v_1^{b,mid}, v_2^{b,large})$, the mechanism allocates item $b$ to bidder $2$. \label{condi-4-ud}
\item Given the valuation profile $I_5 = (v_1^{a,large}, v_2^{a,mid})$, the mechanism allocates item $a$ to bidder $1$. \label{condi-5-ud}
\item Given the valuation profile $I_6 = (v_1^{b,small}, v_2^{b,one})$, the mechanism allocates item $b$ to bidder $1$. \label{condi-6-ud}
\item Given the valuation profile $I_7 = (v_1^{a,one}, v_2^{a,small})$, the mechanism allocates item $a$ to bidder $2$. 
\end{enumerate}
\end{lemma}
The proof of \cref{lemma-ud-instances} is straightforward: if a deterministic mechanism does not satisfy one of the conditions, then due to the fact that $k$ is arbitrarily large, its approximation guarantee is at most $\frac{7}{8}$ with respect to the distribution $\mathcal{D}$.

\subsubsection[Reaching a Contradiction: No Deterministic and OSP Mechanism Succeeds on D]{Reaching a Contradiction: No Deterministic and OSP Mechanism Succeeds on $\mathcal D$}
\label{subsubsec-contradiction-ud}
We now employ \cref{lemma-ud-instances} to prove  \cref{thm-lb-ud}. Let the domain $V_i$ of each bidder consist 
of unit-demand valuations with values in  $\{0,1,\ldots,k^4\}$, where $k$ is an arbitrarily large integer.

Fix a deterministic mechanism $A$ and strategies
$(\mathcal S_1,\mathcal S_2)$ that are individually rational, satisfy no negative transfers for all the valuations  $V_1\times V_2$ and give approximation better than $\frac{7}{8}$
in expectation over the valuation profiles in the distribution $\mathcal D$. 
Denote with  $(f,P_1,P_2)$ be the allocation rule and the payment scheme that $A$ and $(\mathcal S_1,\mathcal S_2)$ realize, and assume towards a contradiction that $A$ and $(\mathcal S_1,\mathcal S_2)$ are obviously strategy-proof.

For the analysis of the deterministic mechanism $A$, we focus on the following subsets of the domains of the valuations:
$
\mathcal{V}_1=\{v_1^{a,one},v_1^{a,mid},v_1^{a,large},v_1^{b,small},v_1^{b,mid},v_1^{b,large},v_1^{both}\}$ and  $\mathcal V_2=\{v_2^{b,one},v_2^{b,mid},v_2^{b,large},\allowbreak v_2^{a,small},v_2^{a,mid}, \allowbreak v_2^{a,large},v_2^{both}\}$.
Observe that there necessarily exists a vertex $u$, and valuations $v_1,v_1' \in \mathcal{V}_1$, and  $v_2,v_2' \in \mathcal{V}_2$ such that $(\mathcal{S}_1(v_1), \mathcal{S}_2(v_2))$ diverge at vertex $u$. This follows from \cref{lemma-ud-instances}, which implies that the mechanism $A$ must output different allocations for different valuation profiles in $\mathcal{V}_1 \times \mathcal{V}_2$. Consequently, not all valuation profiles end up in the same leaf, meaning that divergence must occur at some point.

Let $u$ be the first vertex in the protocol such that 
 $(\mathcal{S}_1(v_1),\mathcal{S}_2(v_2))$ and $(\mathcal{S}_1(v_1'),\mathcal{S}_2(v_2'))$ diverge, i.e., dictate different messages. 
Note that by definition this implies that $u\in Path(\mathcal{S}_1(v_1),\mathcal{S}_2(v_2))\cap Path(\mathcal{S}_1(v_1'),\mathcal{S}_2(v_2'))$ and that either bidder $1$ or bidder $2$ sends different messages for the valuations in $\mathcal{V}_1$ or $\mathcal V_2$, respectively. 
Without loss of generality, we assume that bidder $1$ sends different messages, meaning that there exist $v_1,v_1'\in \mathcal{V}_1$ such that $\mathcal S_1(v_1)$ and $\mathcal S_1(v_1')$ dictate different messages at vertex $u$.
However, the following collection of claims show that since the strategy $\mathcal S_1$ is obviously dominant, it dictates the same message for all the valuations in $\mathcal{V}_1$. Thus, we get a contradiction, which completes the proof of \cref{thm-lb-ud}:
\begin{claim}\label{claim-a-one-mid}
    The strategy $\mathcal S_1$ dictates the same message at vertex $u$ for the valuations $v_1^{a,one}$ and $v_1^{a,mid}$. 
\end{claim}
\begin{claim}\label{claim-a-mid-large}
    The strategy $\mathcal S_1$ dictates the same message at vertex $u$ for the valuations $v_1^{a,mid}$ and $v_1^{a,large}$. 
\end{claim}
\begin{claim}\label{claim-b-small-large}
    The strategy $\mathcal S_1$ dictates the same message at vertex $u$ for the valuations $v_1^{b,small}$ and $v_1^{b,mid}$. 
\end{claim}
\begin{claim}\label{claim-b-mid-large}
    The strategy $\mathcal S_1$ dictates the same message at vertex $u$ for the valuations $v_1^{b,mid}$ and $v_1^{b,large}$. 
\end{claim}
\begin{claim}\label{claim-a-both}
    The strategy $\mathcal S_1$ dictates the same message at vertex $u$ for the valuations $v_1^{both}$ and $v_1^{a,mid}$. 
\end{claim}
\begin{claim}\label{claim-b-both}
    The strategy $\mathcal S_1$ dictates the same message at vertex $u$ for the valuations $v_1^{both}$ and $v_1^{b,mid}$. 
\end{claim}
To prove these claims, we use the following collection of observations about the allocation and the payment scheme of player $1$:    
\begin{lemma}\label{lemma-small-pay-ud}
    The allocation rule $f$ and the payment scheme $P_1$ of bidder $1$ satisfy that:
    \begin{enumerate}
        \item Given the valuation profiles $(v_1^{a,one},v_{2}^{b,one})$ 
        and $(v_1^{a,large},v_{2}^{b,one})$, bidder $1$ wins item $a$ and pays at most $1$.  \label{item-1-ud}
        \item  Given the valuation profile $(v_1^{a,mid},v_2^{a,large})$, bidder $1$ wins a bundle that does not contain item $a$ and pays zero.   \label{item-2-ud}
        \item Given the valuation profiles $(v_1^{b,small},v_2^{b,one})$,
        $(v_1^{b,mid},v_2^{b,one})$ and $(v_1^{b,large},v_2^{b,one})$,  
        bidder $1$ wins item $b$ and pays at most $k$. 
        \label{item-4-ud}
           \item  Given the valuation profile $(v_1^{b,mid},v_2^{b,large})$, bidder $1$ wins a bundle that does not contain item $b$ and pays zero.   \label{item-4.5-ud}
\item Given the valuation profile $(v_1^{both},v_2^{b,one})$, 
bidder $1$ gets a bundle that contains item $a$ and pays at most $1$. \label{item-6-ud}
\item Given the valuation profile $(v_1^{both},v_2^{a,large})$, bidder $1$ does not win item $a$.  If bidder $1$ wins item $b$, then he pays at most $2k+1$. If he wins a bundle that contains neither item $a$ or item $b$, then he pays zero. \label{item-complicated-ud}
    \end{enumerate}
\end{lemma}
\subsubsection[All Valuations in V\_1 Send The Same Message: Proofs of Claims \ref{claim-a-one-mid} to \ref{claim-b-both}]{All Valuations in $\mathcal V_1$ Send The Same Message: Proofs of Claims \ref{claim-a-one-mid} to \ref{claim-b-both}}\label{app-claims-proofs-ud}
We now prove Claims \ref{claim-a-one-mid} to \ref{claim-b-both}, which jointly imply a contradiction. All proofs make extensive use of \cref{lemma-ud-instances}, which analyzes the allocation and the payments of any deterministic mechanisms with the desired properties. All proofs are quite similar to each other, and we write them for completeness. The only claim that requires a more involved case analysis is \cref{claim-b-both}.

\begin{proof}[Proof of \cref{claim-a-one-mid}]
      Note that by Lemma \ref{lemma-small-pay-ud} part \ref{item-1-ud}, $f(v_1^{a,one},v_2^{b,one})$ allocates item $a$ to bidder $1$ and $P_1(v_1^{a,one},v_2^{b,one})\le 1$. Therefore:
\begin{equation}\label{eq-good-leaf1-ud}
 v_1^{a,mid}(f(v_1^{a,one},v_2^{b,one}))-P_1(v_1^{a,one},v_2^{b,one})\ge k^2-1   
\end{equation}
 In contrast, by part \ref{item-2-ud} of Lemma \ref{lemma-small-pay-ud}, the allocation rule  $f(v_1^{a,mid},v_2^{a,large})$ allocates no items to player $1$ and $P_1(v_1^{a,mid},v_2^{a,large})=0$, so:
 \begin{equation}\label{eq-bad-leaf1-ud}
 v_1^{a,mid}(f(v_1^{a,mid},v_2^{a,large}))-P_1(v_1^{a,mid},v_2^{a,large})= 0   
\end{equation}
We remind that $k$ is arbitrarily large,
so combining inequalities (\ref{eq-good-leaf1-ud}) and (\ref{eq-bad-leaf1-ud}) gives:
\begin{equation*}
 v_1^{a,mid}(f(v_1^{a,mid},v_2^{a,large}))-P_1(v_1^{a,mid},v_2^{a,large})< 
v_1^{a,mid}(f(v_1^{a,one},v_2^{b,one}))-P_1(v_1^{a,one},v_2^{b,one})  
\end{equation*}
We remind that vertex $u$ belongs in $Path(\mathcal S_1(v_1^{a,one}),\mathcal S_2(v_2^{b,one}))$ and also in
$Path(\mathcal{S}_1(v_1^{a,mid}),\allowbreak\mathcal{S}_2(v_2^{a,large}))$. Therefore, Lemma \ref{lemma-bad-leaf-good-leaf} gives that the strategy $\mathcal S_1$ dictates the same message for  $v_1^{a,one}$ and $v_1^{a,mid}$ at vertex $u$.
\end{proof}

\begin{proof}[Proof of \cref{claim-a-mid-large}]
     By Lemma \ref{lemma-small-pay-ud} part \ref{item-1-ud}, $f(v_1^{a,large},v_2^{b,one})$ allocates item $a$ to bidder $1$ and $P_1(v_1^{a,large},\allowbreak v_2^{b,one})\le 1$. Therefore:
\begin{equation}\label{eq-good-leaf2-ud}
 v_1^{a,mid}(f(v_1^{a,large},v_2^{b,one}))-P_1(v_1^{a,large},v_2^{b,one})\ge k^2-1   
\end{equation}
Whereas by part \ref{item-2-ud} of Lemma \ref{lemma-small-pay-ud}: 
 \begin{equation}\label{eq-bad-leaf2-ud}
 v_1^{a,mid}(f(v_1^{a,mid},v_2^{a,large}))-P_1(v_1^{a,mid},v_2^{a,large})=0   
\end{equation}
Combining the inequalities (\ref{eq-good-leaf2-ud}) and (\ref{eq-bad-leaf2-ud}) gives:
\begin{equation*}
 v_1^{a,mid}(f(v_1^{a,mid},v_2^{a,large}))-P_1(v_1^{a,mid},v_2^{a,large})< 
v_1^{a,mid}(f(v_1^{a,large},v_2^{b,one}))-P_1(v_1^{a,large},v_2^{b,one})
\end{equation*}
Note that vertex $u$ belongs in $Path(\mathcal S_1(v_1^{a,large}),\mathcal S_2(v_2^{b,one}))$ and also in
$Path(\mathcal{S}_1(v_1^{a,mid}),\allowbreak\mathcal{S}_2(v_2^{a,large}))$. Therefore, Lemma \ref{lemma-bad-leaf-good-leaf} gives that the strategy $\mathcal S_1$ dictates the same message for  $v_1^{a,mid}$ and $v_1^{a,large}$ at vertex $u$.
\end{proof}

\begin{proof}[Proof of \cref{claim-b-small-large}]
        By Lemma \ref{lemma-small-pay-ud} part \ref{item-4-ud}, $f(v_1^{b,small},v_2^{b,one})$ allocates item $b$ to bidder $1$ and $P_1(v_1^{b,small},\allowbreak v_2^{b,one})\le k$. Therefore:
\begin{equation}\label{eq-good-leaf3-ud}
 v_1^{b,mid}(f(v_1^{b,small},v_2^{b,one}))-P_1(v_1^{b,small},v_2^{b,one})\ge k^2-k   
\end{equation}
Whereas by part \ref{item-4.5-ud} of Lemma \ref{lemma-small-pay-ud}: 
 \begin{equation}\label{eq-bad-leaf3-ud}
 v_1^{b,mid}(f(v_1^{b,mid},v_2^{b,large}))-P_1(v_1^{b,mid},v_2^{b,large})= 0   
\end{equation}
Combining inequalities (\ref{eq-good-leaf3-ud}) and (\ref{eq-bad-leaf3-ud}) gives:
\begin{equation*}
 v_1^{b,mid}(f(v_1^{b,mid},v_2^{b,large}))-P_1(v_1^{b,mid},v_2^{b,large})< 
v_1^{b,mid}(f(v_1^{b,small},v_2^{b,one}))-P_1(v_1^{b,small},v_2^{b,one})
\end{equation*}
Note that vertex $u$ belongs in $Path(\mathcal S_1(v_1^{b,small}),\mathcal S_2(v_2^{b,one}))$ and also in
$Path(\mathcal{S}_1(v_1^{b,mid}),\allowbreak\mathcal{S}_2(v_2^{b,large}))$. Therefore, Lemma \ref{lemma-bad-leaf-good-leaf} gives that the strategy $\mathcal S_1$ dictates the same message for  $v_1^{b,mid}$ and $v_1^{b,small}$ at vertex $u$.
\end{proof}
\begin{proof}[Proof of \cref{claim-b-mid-large}]
          By Lemma \ref{lemma-small-pay-ud} part \ref{item-4-ud}, $f(v_1^{b,large},v_2^{b,one})$ allocates item $b$ to bidder $1$ and $P_1(v_1^{b,large},\allowbreak v_2^{b,one})\le k$. Therefore:
\begin{equation}\label{eq-good-leaf4-ud}
 v_1^{b,mid}(f(v_1^{b,large},v_2^{b,one}))-P_1(v_1^{b,large},v_2^{b,one})\ge k^2-k   
\end{equation}
Whereas by part \ref{item-4.5-ud} of Lemma \ref{lemma-small-pay-ud}: 
 \begin{equation}\label{eq-bad-leaf4-ud}
 v_1^{b,mid}(f(v_1^{b,mid},v_2^{b,large}))-P_1(v_1^{b,mid},v_2^{b,large})= 0   
\end{equation}
Combining inequalities (\ref{eq-good-leaf4-ud}) and (\ref{eq-bad-leaf4-ud}) gives:
\begin{equation*}
 v_1^{b,mid}(f(v_1^{b,mid},v_2^{b,large}))-P_1(v_1^{b,mid},v_2^{b,large})< 
v_1^{b,mid}(f(v_1^{b,large},v_2^{b,one}))-P_1(v_1^{b,large},v_2^{b,one})
\end{equation*}
Note that vertex $u$ belongs in $Path(\mathcal S_1(v_1^{b,large}),\mathcal S_2(v_2^{b,one}))$ and also in
$Path(\mathcal{S}_1(v_1^{b,mid}),\allowbreak\mathcal{S}_2(v_2^{b,large}))$. Therefore, Lemma \ref{lemma-bad-leaf-good-leaf} gives that the strategy $\mathcal S_1$ dictates the same message for  $v_1^{b,mid}$ and $v_1^{b,large}$ at vertex $u$.
\end{proof}
\begin{proof}[Proof of \cref{claim-a-both}]
            By Lemma \ref{lemma-small-pay-ud} part \ref{item-6-ud}, $f(v_1^{both},v_2^{b,one})$ allocates item $a$ to bidder $1$ and $P_1(v_1^{both},\allowbreak v_2^{b,one})\le 1$. Therefore:
\begin{equation}\label{eq-good-leaf5-ud}
 v_1^{a,mid}(f(v_1^{both},v_2^{b,one}))-P_1(v_1^{both},v_2^{b,one})\ge k^2-1   
\end{equation}
Whereas by part \ref{item-2-ud} of Lemma \ref{lemma-small-pay-ud}: 
 \begin{equation}\label{eq-bad-leaf5-ud}
 v_1^{a,mid}(f(v_1^{a,mid},v_2^{a,large}))-P_1(v_1^{a,mid},v_2^{a,large})=0   
\end{equation}
Combining the inequalities (\ref{eq-good-leaf5-ud}) and (\ref{eq-bad-leaf5-ud}) gives:
\begin{equation*}
 v_1^{a,mid}(f(v_1^{a,mid},v_2^{a,large}))-P_1(v_1^{a,mid},v_2^{a,large})< 
 v_1^{a,mid}(f(v_1^{both},v_2^{b,one}))-P_1(v_1^{both},v_2^{b,one})
\end{equation*}
Note that vertex $u$ belongs in $Path(\mathcal S_1(v_1^{both}),\mathcal S_2(v_2^{b,one}))$ and also in
$Path(\mathcal{S}_1(v_1^{a,mid}),\allowbreak\mathcal{S}_2(v_2^{a,large}))$. Therefore, Lemma \ref{lemma-bad-leaf-good-leaf} gives that the strategy $\mathcal S_1$ dictates the same message for  $v_1^{a,mid}$ and $v_1^{both}$ at vertex $u$.
\end{proof}
\begin{proof}[Proof of \cref{claim-b-both}]
    Note that by \cref{lemma-small-pay-ud} part \ref{item-complicated-ud}, given the valuation profile $(v_1^{both},v_2^{a,large})$, bidder $1$ cannot win item $a$ so we consider the following two cases: the case where he wins item $b$ and the case where he wins neither of these items.

    Assume that bidder $1$ wins item $b$ given the valuation profile $(v_1^{both},v_2^{a,large})$. Note that in this case, \cref{lemma-small-pay-ud} part \ref{item-complicated-ud} also implies that:
    \begin{equation}\label{eq-good-leaf6-ud}
 v_1^{b,mid}(f(v_1^{both},v_2^{a,large}))-P_1(v_1^{both},v_2^{a,large})\ge k^2 -2k-1   
\end{equation}
Whereas by part \ref{item-4.5-ud} of Lemma \ref{lemma-small-pay-ud}: 
 \begin{equation}\label{eq-bad-leaf6-ud}
 v_1^{b,mid}(f(v_1^{b,mid},v_2^{b,large}))-P_1(v_1^{b,mid},v_2^{b,large})= 0   
\end{equation}
Combining the inequalities (\ref{eq-good-leaf6-ud}) and (\ref{eq-bad-leaf6-ud}) gives:
\begin{equation*}
 v_1^{b,mid}(f(v_1^{b,mid},v_2^{b,large}))-P_1(v_1^{b,mid},v_2^{b,large})< 
 v_1^{b,mid}(f(v_1^{both},v_2^{a,large}))-P_1(v_1^{both},v_2^{a,large})
\end{equation*}
Note that vertex $u$ belongs in $Path(\mathcal S_1(v_1^{both}),\mathcal S_2(v_2^{a,large}))$ and also in
$Path(\mathcal{S}_1(v_1^{a,mid}),\allowbreak\mathcal{S}_2(v_2^{a,large}))$. Therefore, Lemma \ref{lemma-bad-leaf-good-leaf} gives that the strategy $\mathcal S_1$ dictates the same message for  $v_1^{both}$ and $v_1^{b,mid}$ at vertex $u$, which concludes this case.

For the latter case, where bidder $1$ gets a bundle that contains neither item $a$ nor item $b$ given the valuation profile $(v_1^{both},v_2^{a,large})$, by \cref{lemma-small-pay-ud} part \ref{item-complicated-ud}: 
\begin{equation}\label{eq-bad-leaf7-ud}
 v_1^{both}(f(v_1^{both},v_2^{a,large}))-P_1(v_1^{both},v_2^{a,large})=0   
\end{equation}
Whereas by part \ref{item-4-ud} of Lemma \ref{lemma-small-pay-ud}: 
 \begin{equation}\label{eq-good-leaf7-ud}
 v_1^{both}(f(v_1^{b,mid},v_2^{b,one}))-P_1(v_1^{b,mid},v_2^{b,one}) \ge 2k+1-k=k+1   
\end{equation}
Combining (\ref{eq-bad-leaf7-ud}) and  (\ref{eq-good-leaf7-ud})  gives:
\begin{equation*}
 v_1^{both}(f(v_1^{both},v_2^{a,large}))-P_1(v_1^{both},v_2^{a,large})< 
 v_1^{both}(f(v_1^{b,mid},v_2^{b,one}))-P_1(v_1^{b,mid},v_2^{b,one})
\end{equation*}
Note that vertex $u$ belongs in $Path(\mathcal S_1(v_1^{both}),\mathcal S_2(v_2^{a,large}))$ and also in
$Path(\mathcal{S}_1(v_1^{b,mid}),\allowbreak\mathcal{S}_2(v_2^{b,one}))$. Therefore, Lemma \ref{lemma-bad-leaf-good-leaf} gives that the strategy $\mathcal S_1$ dictates the same message for  $v_1^{both}$ and $v_1^{b,mid}$ at vertex $u$, which resolves the second case and completes the proof.
\end{proof}

\subsubsection{Proof of Lemma \ref{lemma-small-pay-ud}:
Observations About The Mechanism} \label{subsubsec-prop-alloc-ud}
The proof is a direct consequence of the approximation guarantee of the mechanism and the fact that it is 
obviously strategy-proof (and thus also dominant-strategy incentive compatible) and satisfies individual rationality and no negative transfers. 
We write the proof  for the sake of completeness.

Throughout the proof, we say that an item $x$ is \emph{valuable} for a valuation $v$ if $v(\{x\})>0$.

\begin{proof}[Proof of \cref{lemma-small-pay-ud}]
    
For part \ref{item-1-ud}, note that because of individual rationality:
\begin{equation}\label{eq-1-ud}
 v_1^{a,one}(f(v_1^{a,one},v_2^{b,one}))-P_1(v_1^{a,one},v_2^{b,one})\ge 0   
\end{equation}
We remind that by \cref{lemma-ud-instances} part \ref{condi-1-ud}, 
the allocation rule $f$ allocates item $a$ to bidder $1$ given $(v_1^{one},v_2^{one})$, so:
\begin{equation}\label{eq-2-ud}
 v_1^{a,one}(f(v_1^{a,one},v_2^{b,one}))=1   
\end{equation}
Combining (\ref{eq-1-ud}) and (\ref{eq-2-ud}) proves the part \ref{item-1-ud} for the valuation profile $(v_1^{a,one},v_2^{b,one})$.
Observe that it also implies that: 
\begin{equation*}\label{mid-one-eq-ud}
v_1^{a,large}(f(v_1^{a,one},v_2^{b,one}))-P_1(v_1^{a,one},v_2^{b,one})\ge k^4-1
\end{equation*}
Combining this inequality with the fact that the allocation rule $f$ and the payment scheme $P_1$ are realized by a dominant-strategy mechanism gives that:
\begin{equation}\label{eq-part1-unit}
\begin{aligned}
    v_1^{a,large}(f(v_1^{a,large},v_2^{b,one}))-P_1(v_1^{a,large},v_2^{b,one})&\ge v_1^{a,large}(f(v_1^{a,one},v_2^{b,one}))-P_1(v_1^{a,one},v_2^{b,one}) \\
    &\ge k^4-1 
\end{aligned}  
\end{equation}
Note that the property of no negative transfers implies that $P_1(v_1^{a,large},v_2^{b,one})\ge 0$, so
$v_1^{a,large}(f(v_1^{a,large},\allowbreak v_2^{b,one}))\ge k^4-1$. Since only item $a$ is valuable for  $v_1^{a,large}$, we can deduce that player $1$ wins it given the valuation profile $(v_1^{a,large}\allowbreak,v_2^{b,one})$, which further implies that in fact:
\begin{equation}\label{eq-part1-another}
v_1^{a,large}(f(v_1^{a,large}\allowbreak,v_2^{b,one}))= k^4    
\end{equation}
Now, combining (\ref{eq-part1-unit}) and (\ref{eq-part1-another}) gives that $P_1(v_1^{a,large},v_2^{b,one})\le 1$, which completes the proof for $(v_1^{a,large},v_1^{b,one})$. 

For part \ref{item-2-ud}, observe that by  \cref{lemma-ud-instances} part \ref{condi-2-ud}, given $(v_1^{a,mid},v_2^{a,large})$, bidder $2$ gets item $a$, so clearly bidder $1$ does not get it. Since only item $a$ is valuable for $v_1^{a,mid}$, we get that $v_1^{a,mid}(f(v_1^{a,mid},v_2^{a,large}))\allowbreak=0$, so by individual rationality $P_1(v_1^{a,mid},v_2^{a,large})\le 0$. Due to no negative transfers, we get that $P_1(v_1^{a,mid},v_2^{a,large})=0$, which completes the proof of this part. 

We now prove part \ref{item-4-ud}. We begin by proving it for the valuation profile $(v_1^{b,small},v_2^{b,one})$. Note that by \cref{lemma-ud-instances} part \ref{condi-6-ud}, given this valuation profile, bidder $1$ wins item $b$. Combining this fact with the fact that the mechanism satisfies individual rationality implies that $P_1(v_1^{b,small},v_2^{b,one})\le k$, as needed. Showing that the same goes for $(v_1^{b,mid},v_2^{b,one})$ and $(v_1^{b,large},v_2^{b,one})$ is analogous to the proof of part \ref{item-1-ud} for the valuation profile $(v_1^{a,mid},v_2^{b,one})$ above.

To prove \ref{item-4.5-ud}, note that by \cref{lemma-ud-instances} part \ref{condi-4-ud}, bidder $1$ does not win item $b$ given $(v_1^{b,mid},v_2^{b,large})$. Thus, he has to pay at most zero due to individual rationality, and the property of no negative transfers implies that $P_1(v_1^{b,mid},v_2^{b,large})=0$, as needed.  

For part \ref{item-6-ud}, we first show that bidder $1$ wins a bundle that contains item $a$ given $(v_1^{both},v_2^{b,one})$. Note that since the allocation rule $f$ and the payment scheme $P_1$ are realized by a dominant-strategy mechanism, we have that:
\begin{equation}\label{eq-part6-unit}
\begin{aligned}
    v_1^{both}(f(v_1^{both},v_2^{b,one}))-P_1(v_1^{both},v_2^{b,one})&\ge v_1^{both}(f(v_1^{a,one},v_2^{b,one}))-P_1(v_1^{a,one},v_2^{b,one}) \\
    &\ge 2k+2 &\text{(by part \ref{item-1-ud})}
\end{aligned}    
\end{equation}
Combining (\ref{eq-part6-unit}) with the property of no negative transfers implies that $v_1^{both}(f(v_1^{both},v_2^{b,one}))\ge 2k+2$. Thus, given $(v_1^{both},v_2^{b,one})$, bidder $1$ necessarily gets a bundle that contains item $a$.

For the upper bound on the payment, note that in fact $v_1^{both}(f(v_1^{both},v_2^{b,one}))=2k+3$.  
Combining it with inequality (\ref{eq-part6-unit}) gives that $P_1(v_1^{both},v_2^{b,one})\le 1$. By that, we  complete the proof of part \ref{item-6-ud}. 

We are now finally ready to wrap up by proving part \ref{item-complicated-ud}. We begin by showing that given $(v_1^{both},v_2^{a,large})$, bidder $1$ does not win item $a$. This is due to weak monotonicity. Formally, we remind that by part \ref{item-2-ud}, given $(v_1^{a,mid},v_2^{a,large})$, bidder $1$ wins a bundle $S^{mid}$
that does not contain $a$. Thus, if $f(v_1^{both},v_2^{a,large})$ allocates to bidder $1$ 
a bundle $S^{both}$ that contains item $a$, then $v_1^{a,mid}(S^{both})-v_1^{a,mid}(S^{mid})=k^2$ whereas $v_1^{both}(S^{both})-v_1^{both}(S^{mid})=2k+3$, so $f$ is not weakly monotone. Since $f$ and $P_1$ realize a dominant-strategy mechanism, \cref{wmon-lemma} gives that $f$ has to be weakly-monotone, so we get a contradiction. Thus, $f(v_1^{both},v_2^{a,large})$ has to allocate bidder $1$ a bundle that does not contain item $a$. 

The bounds on the payments are a straightforward implication of individual rationality and no negative transfers. If $f(v_1^{both},v_2^{a,large})$ allocates item $b$ to bidder $1$ then the payment is at most $2k+1$, and if it allocates to bidder $1$ no valuable items, then the payment has to be at most zero. Due to no negative transfers, it is zero exactly.   
\end{proof}

\subsection{Proof of Theorem \ref{thm-lb-add}: An Impossibility for Additive Bidders} \label{lb-add-proof-place}
The proof follows the same structure as in \cref{thm-mua-sm-lb} and \cref{thm-lb-ud}: roughly speaking, we describe a distribution $\mathcal{D}$, show that it is “hard” for deterministic mechanisms, and then use Yao's Lemma to deduce hardness for randomized mechanisms. In particular, this proof is very similar to the proof of \cref{thm-lb-ud} in \cref{lb-ud-proof-place}, and we write both for the sake of completeness. 

The proof outline is as follows. In \cref{subsubsec-description-add}, 
we describe the  hard distribution $\mathcal D$. 
In \cref{subsubsec-performance-add} we analyze the allocation and payments of a deterministic \textquote{good} mechanism 
and explain why these properties  imply that  a \textquote{good} mechanism does not exist in  \cref{subsubsec-contradiction-add}.
We defer technical proofs to \cref{app-claims-proofs-add,subsubsec-prop-alloc-add}. 

We prove for the case of two bidders and two items, but the proof extends to any number of bidders and any number of items by adding bidders with the all-zero valuation and assuming that the bidders in our construction have zero values for the additional items.

\subsubsection[Construction of a "Hard" Distribution D]{Construction of a \textquote{Hard} Distribution $\mathcal D$} \label{subsubsec-description-add}
    To define the probability distribution over the valuation profiles, we name the two items $a$ and $b$, and let $k$ be an arbitrarily large number. Note that since these valuations are additive, we can fully describe them by specifying their value for each item.
    Consider the following valuations:  
\[
v_1^{{a,one}}(x) = 
\begin{cases}
1 & x=a,\\
0 & \text{otherwise.}
\end{cases} \quad 
v_2^{{b,one}}(x) = 
\begin{cases}
1 & x=b,\\
0 & \text{otherwise.}
\end{cases}
\]

\[
v_1^{{a,mid}}(x) = 
\begin{cases}
k^2 & x=a,\\
0 & \text{otherwise.}
\end{cases} \quad 
v_2^{{b,mid}}(x) = 
\begin{cases}
k^2 & x=b,\\
0 & \text{otherwise.}
\end{cases}
\]

\[
v_1^{{a,large}}(x) = 
\begin{cases}
k^4 & x=a,\\
0 & \text{otherwise.}
\end{cases} \quad 
v_2^{{b,large}}(x) = 
\begin{cases}
k^4 & x=b,\\
0 & \text{otherwise.}
\end{cases}
\]

\[
v_1^{{b,small}}(x) = 
\begin{cases}
k & x=b,\\
0 & \text{otherwise.}
\end{cases} \quad 
v_2^{{a,small}}(x) = 
\begin{cases}
k & x=a,\\
0 & \text{otherwise.}
\end{cases}
\]

\[
v_1^{{b,mid}}(x) = 
\begin{cases}
k^2 & x=b,\\
0 & \text{otherwise.}
\end{cases} \quad 
v_2^{{a,mid}}(x) = 
\begin{cases}
k^2 & x=a,\\
0 & \text{otherwise.}
\end{cases}
\]

\[
v_1^{{b,large}}(x) = 
\begin{cases}
k^4 & x=b,\\
0 & \text{otherwise.}
\end{cases} \quad 
v_2^{{a,large}}(x) = 
\begin{cases}
k^4 & x=a,\\
0 & \text{otherwise.}
\end{cases}
\]

\[
v_1^{{both}}(x) = 
\begin{cases}
2k+3 & x=a,\\
2k+1 & x=b,\\
0 & \text{otherwise.}
\end{cases} \quad 
v_2^{{both}}(x) = 
\begin{cases}
2k+3 & x=b,\\
2k+1 & x=a,\\
0 & \text{otherwise.}
\end{cases}
\]
Note that all the valuations except for $v_1^{both}$ and $v_2^{both}$ are unit-demand and additive simultaneously. 
Consider the following valuation profiles: 
\begin{align*}
 I_1 = (v_1^{{a,one}},& v_2^{{b,one}}), \quad  I_2 = (v_1^{a,mid}, v_2^{{a,large}}),  \quad I_3 = (v_1^{b,large}, v_2^{b,mid}), \quad
 I_4 = (v_1^{b,mid}, v_2^{b,large}) \\ & I_5 = (v_1^{a,large}, v_2^{a,mid}) \quad I_6 = (v_1^{b,small}, v_2^{b,one}) \quad I_7 =(v_1^{a,one}, v_2^{a,small}) 
\end{align*}
Let $\mathcal D$ be the distribution over valuation profiles where the probability of the valuation profile $I_1$ is $\frac{1}{4}$, and the probability of the valuation profiles $I_2,I_3,I_4,I_5,I_6$ and $I_7$ is $\frac{1}{8}$ each. 

\subsubsection[The Performance of the Deterministic Mechanism A on the "Hard" Distribution D]{The Performance of the Deterministic Mechanism $A$ on the \textquote{Hard} Distribution $\mathcal D$}\label{subsubsec-performance-add}
We remind that our goal is to show that no deterministic mechanism that satisfies all of the desired properties extracts more than $\frac{7}{8}$ of the optimal welfare. For that, we begin by observing that:
\begin{lemma}\label{lemma-add-instances}
Every deterministic mechanism that has approximation better than $\frac{8}{7}$ necessarily satisfies all of the following conditions:
\begin{enumerate}
    \item Given the valuation profile $I_1 = (v_1^{\text{a,one}}, v_2^{\text{b,one}})$, the mechanism
    allocates item $a$ to player $1$ and item $b$ to player $2$. \label{condi-1-add}
\item Given the valuation profile $I_2 = (v_1^{a,mid}, v_2^{a,large})$, the mechanism allocates item $a$ to bidder $2$. \label{condi-2-add}

    \item Given the valuation profile $I_3 = (v_1^{b,large}, v_2^{b,mid})$, the mechanism allocates item $b$  to bidder $1$.
    \label{condi-3-add}
\item Given the valuation profile $I_4 = (v_1^{b,mid}, v_2^{b,large})$, the mechanism allocates item $b$ to bidder $2$. \label{condi-4-add}
\item Given the valuation profile $I_5 = (v_1^{a,large}, v_2^{a,mid})$, the mechanism allocates item $a$ to bidder $1$. \label{condi-5-add}
\item Given the valuation profile $I_6 = (v_1^{b,small}, v_2^{b,one})$, the mechanism allocates item $b$ to bidder $1$. \label{condi-6-add}
\item Given the valuation profile $I_7 = (v_1^{a,one}, v_2^{a,small})$, the mechanism allocates item $a$ to bidder $2$. 
\end{enumerate}
\end{lemma}
The proof of \cref{lemma-add-instances} is straightforward and identical to the proof of \cref{lemma-ud-instances}: if a deterministic mechanism does not satisfy one of conditions, then the fact that $k$ is arbitrarily large implies that its approximation guarantee  is at most $\frac{8}{7}$ with respect to the distribution $\mathcal D$. 
\subsubsection[Reaching a Contradiction: No Deterministic and OSP Mechanism Succeeds on D]{Reaching a Contradiction: No Deterministic and OSP Mechanism Succeeds on $\mathcal D$}
\label{subsubsec-contradiction-add}
We now employ \cref{lemma-add-instances} to prove  \cref{thm-lb-add}. Let the domain $V_i$ of each bidder consist 
of additive valuations with values in  $\{0,1,\ldots,k^4\}$, where $k$ is an arbitrarily large integer.

Fix a deterministic mechanism $A$ and strategies
$(\mathcal S_1,\mathcal S_2)$ that are individually rational, satisfy no negative transfers for all the valuations  $V_1\times V_2$ and give approximation better than $\frac{7}{8}$
in expectation over the valuation profiles in the distribution $\mathcal D$. 
Denote with  $(f,P_1,P_2)$ be the allocation rule and the payment scheme that $A$ and $(\mathcal S_1,\mathcal S_2)$ realize, and assume towards a contradiction that $A$ and $(\mathcal S_1,\mathcal S_2)$ are obviously strategy-proof.

For the analysis of the deterministic mechanism $A$, we focus on the following subsets of the domains of the valuations:
$
\mathcal{V}_1=\{v_1^{a,one},v_1^{a,mid},v_1^{a,large},v_1^{b,small},v_1^{b,mid},v_1^{b,large},v_1^{both}\}$ and  $\mathcal V_2=\{v_2^{b,one},v_2^{b,mid},v_2^{b,large},\allowbreak v_2^{a,small},v_2^{a,mid}, \allowbreak v_2^{a,large},v_2^{both}\}$.
Observe that there necessarily exists a vertex $u$, and valuations $v_1,v_1' \in \mathcal{V}_1$, and  $v_2,v_2' \in \mathcal{V}_2$ such that $(\mathcal{S}_1(v_1), \mathcal{S}_2(v_2))$ diverge at vertex $u$. This follows from \cref{lemma-add-instances}, which implies that the mechanism $A$ must output different allocations for different valuation profiles in $\mathcal{V}_1 \times \mathcal{V}_2$. Consequently, not all valuation profiles end up in the same leaf, meaning that divergence must occur at some point.

Let $u$ be the first vertex in the protocol such that 
 $(\mathcal{S}_1(v_1),\mathcal{S}_2(v_2))$ and $(\mathcal{S}_1(v_1'),\mathcal{S}_2(v_2'))$ diverge, i.e., dictate different messages. 
Note that by definition this implies that $u\in Path(\mathcal{S}_1(v_1),\mathcal{S}_2(v_2))\cap Path(\mathcal{S}_1(v_1'),\mathcal{S}_2(v_2'))$ and that either bidder $1$ or bidder $2$ sends different messages for the valuations in $\mathcal{V}_1$ or $\mathcal V_2$, respectively. 
Without loss of generality, we assume that bidder $1$ sends different messages, meaning that there exist $v_1,v_1'\in \mathcal{V}_1$ such that $\mathcal S_1(v_1)$ and $\mathcal S_1(v_1')$ dictate different messages at vertex $u$.
However, the following collection of claims show that since the strategy $\mathcal S_1$ is obviously dominant, it dictates the same message for all the valuations in $\mathcal{V}_1$. Thus, we get a contradiction, which completes the proof of \cref{thm-lb-add}:
\begin{claim}\label{claim-a-one-mid-add}
    The strategy $\mathcal S_1$ dictates the same message at vertex $u$ for the valuations $v_1^{a,one}$ and $v_1^{a,mid}$. 
\end{claim}
\begin{claim}\label{claim-a-mid-large-add}
    The strategy $\mathcal S_1$ dictates the same message at vertex $u$ for the valuations $v_1^{a,mid}$ and $v_1^{a,large}$. 
\end{claim}
\begin{claim}\label{claim-b-small-large-add}
    The strategy $\mathcal S_1$ dictates the same message at vertex $u$ for the valuations $v_1^{b,small}$ and $v_1^{b,mid}$. 
\end{claim}
\begin{claim}\label{claim-b-mid-large-add}
    The strategy $\mathcal S_1$ dictates the same message at vertex $u$ for the valuations $v_1^{b,mid}$ and $v_1^{b,large}$. 
\end{claim}
\begin{claim}\label{claim-a-both-add}
    The strategy $\mathcal S_1$ dictates the same message at vertex $u$ for the valuations $v_1^{both}$ and $v_1^{a,mid}$. 
\end{claim}
\begin{claim}\label{claim-b-both-add}
    The strategy $\mathcal S_1$ dictates the same message at vertex $u$ for the valuations $v_1^{both}$ and $v_1^{b,mid}$. 
\end{claim}
To prove these claims, we use the following observations about the allocation and the payment scheme of player $1$:    
\begin{lemma}\label{lemma-small-pay-add}
    The allocation rule $f$ and the payment scheme $P_1$ of bidder $1$ satisfy that:
    \begin{enumerate}
        \item Given the valuation profiles $(v_1^{a,one},v_{2}^{b,one})$ and $(v_1^{a,large},v_{2}^{b,one})$, bidder $1$ wins item $a$ and pays at most $1$.  \label{item-1-add}
        \item  Given the valuation profile $(v_1^{a,mid},v_2^{a,large})$, bidder $1$ wins a bundle that does not contain item $a$ and pays zero.   \label{item-2-add}
        \item Given the valuation profiles $(v_1^{b,small},v_2^{b,one})$,
        $(v_1^{b,mid},v_2^{b,one})$ and $(v_1^{b,large},v_2^{b,one})$,  
        bidder $1$ wins item $b$ and pays at most $k$. 
        \label{item-4-add}
           \item  Given the valuation profile $(v_1^{b,mid},v_2^{b,large})$, bidder $1$ wins a bundle that does not contain item $b$ and pays zero.   \label{item-4.5-add}
\item Given the valuation profile $(v_1^{both},v_2^{b,one})$, 
bidder $1$ gets a bundle that contains item $a$ and pays at most $4k+4$. \label{item-6-add}
\item Given the valuation profile $(v_1^{both},v_2^{a,large})$, bidder $1$ does not win item $a$.  If bidder $1$ wins item $b$, then he pays at most $2k+1$. If he wins a bundle that contains neither item $a$ or item $b$, then he pays zero. \label{item-complicated-add}
    \end{enumerate}
\end{lemma}
\subsubsection[All Valuations in V\_1 Send The Same Message: Proofs of Claims \ref{claim-a-one-mid-add} to \ref{claim-b-both-add}]{All Valuations in $\mathcal V_1$ Send The Same Message: Proofs of Claims \ref{claim-a-one-mid-add} to \ref{claim-b-both-add}}\label{app-claims-proofs-add}
We will now prove only \cref{claim-a-both-add} and \cref{claim-b-both-add}, since the proofs of the rest of the claims in fact appear in \cref{app-claims-proofs-ud}: the proof of \cref{claim-a-one-mid-add} is identical to the proof of \cref{claim-a-one-mid}, and the same goes for the proofs of  \cref{claim-a-mid-large-add} and \cref{claim-a-mid-large}, 
the proof of \cref{claim-b-small-large-add} and  \cref{claim-b-small-large} and the proof of \cref{claim-b-mid-large-add} and \cref{claim-b-mid-large}.


\begin{proof}[Proof of \cref{claim-a-both-add}]
            By Lemma \ref{lemma-small-pay-add} part \ref{item-6-add}, $f(v_1^{both},v_2^{b,one})$ allocates item $a$ to bidder $1$ and $P_1(v_1^{both},\allowbreak v_2^{b,one})\le 4k+4$. Therefore:
\begin{equation}\label{eq-good-leaf5-add}
 v_1^{a,mid}(f(v_1^{both},v_2^{b,one}))-P_1(v_1^{both},v_2^{b,one})\ge k^2-4k-4   
\end{equation}
Whereas by part \ref{item-2-add} of Lemma \ref{lemma-small-pay-add}: 
 \begin{equation}\label{eq-bad-leaf5-add}
 v_1^{a,mid}(f(v_1^{a,mid},v_2^{a,large}))-P_1(v_1^{a,mid},v_2^{a,large})=0   
\end{equation}
Combining the inequalities (\ref{eq-good-leaf5-add}) and (\ref{eq-bad-leaf5-add}) gives:
\begin{equation*}
 v_1^{a,mid}(f(v_1^{a,mid},v_2^{a,large}))-P_1(v_1^{a,mid},v_2^{a,large})< 
 v_1^{a,mid}(f(v_1^{both},v_2^{b,one}))-P_1(v_1^{both},v_2^{b,one})
\end{equation*}
Note that vertex $u$ belongs in $Path(\mathcal S_1(v_1^{both}),\mathcal S_2(v_2^{b,one}))$ and also in
$Path(\mathcal{S}_1(v_1^{a,mid}),\allowbreak\mathcal{S}_2(v_2^{a,large}))$. Therefore, Lemma \ref{lemma-bad-leaf-good-leaf} gives that the strategy $\mathcal S_1$ dictates the same message for  $v_1^{a,mid}$ and $v_1^{both}$ at vertex $u$.
\end{proof}
\begin{proof}[Proof of \cref{claim-b-both-add}]
    Note that by \cref{lemma-small-pay-add} part \ref{item-complicated-add}, given the valuation profile $(v_1^{both},v_2^{a,large})$, bidder $1$ cannot win item $a$ so we consider the following two cases: the case where he wins item $b$ and the case where he wins neither of these items.

    Assume that bidder $1$ wins item $b$ given the valuation profile $(v_1^{both},v_2^{a,large})$. Note that in this case, \cref{lemma-small-pay-add} part \ref{item-complicated-add} also implies that:
    \begin{equation}\label{eq-good-leaf6-add}
 v_1^{b,mid}(f(v_1^{both},v_2^{a,large}))-P_1(v_1^{both},v_2^{a,large})\ge k^2 -2k-1   
\end{equation}
Whereas by part \ref{item-4.5-add} of Lemma \ref{lemma-small-pay-add}: 
 \begin{equation}\label{eq-bad-leaf6-add}
 v_1^{b,mid}(f(v_1^{b,mid},v_2^{b,large}))-P_1(v_1^{b,mid},v_2^{b,large})= 0   
\end{equation}
Combining the inequalities (\ref{eq-good-leaf6-add}) and (\ref{eq-bad-leaf6-add}) gives:
\begin{equation*}
 v_1^{b,mid}(f(v_1^{b,mid},v_2^{b,large}))-P_1(v_1^{b,mid},v_2^{b,large})< 
 v_1^{b,mid}(f(v_1^{both},v_2^{a,large}))-P_1(v_1^{both},v_2^{a,large})
\end{equation*}
Note that vertex $u$ belongs in $Path(\mathcal S_1(v_1^{both}),\mathcal S_2(v_2^{a,large}))$ and also in
$Path(\mathcal{S}_1(v_1^{a,mid}),\allowbreak\mathcal{S}_2(v_2^{a,large}))$. Therefore, Lemma \ref{lemma-bad-leaf-good-leaf} gives that the strategy $\mathcal S_1$ dictates the same message for  $v_1^{both}$ and $v_1^{b,mid}$ at vertex $u$, which completing this case.

For the latter case, where bidder $1$ gets a bundle that contains neither item $a$ nor item $b$ given the valuation profile $(v_1^{both},v_2^{a,large})$, note that for \cref{lemma-small-pay-add} part \ref{item-complicated-add} implies that:
\begin{equation}\label{eq-bad-leaf7-add}
 v_1^{both}(f(v_1^{both},v_2^{a,large}))-P_1(v_1^{both},v_2^{a,large})=0   
\end{equation}
Whereas by part \ref{item-4-add} of Lemma \ref{lemma-small-pay-add}: 
 \begin{equation}\label{eq-good-leaf7-add}
 v_1^{both}(f(v_1^{b,mid},v_2^{b,one}))-P_1(v_1^{b,mid},v_2^{b,one}) \ge 2k+1-k=k+1 
\end{equation}
Combining the inequalities (\ref{eq-bad-leaf7-add}) and (\ref{eq-good-leaf7-add}) gives:
\begin{equation*}
 v_1^{both}(f(v_1^{both},v_2^{a,large}))-P_1(v_1^{both},v_2^{a,large})< 
 v_1^{both}(f(v_1^{b,mid},v_2^{b,one}))-P_1(v_1^{b,mid},v_2^{b,one})
\end{equation*}
Note that vertex $u$ belongs in $Path(\mathcal S_1(v_1^{both}),\mathcal S_2(v_2^{a,large}))$ and also in
$Path(\mathcal{S}_1(v_1^{b,mid}),\allowbreak\mathcal{S}_2(v_2^{b,one}))$. Therefore, Lemma \ref{lemma-bad-leaf-good-leaf} gives that the strategy $\mathcal S_1$ dictates the same message for  $v_1^{both}$ and $v_1^{b,mid}$ at vertex $u$, which solves the second case and completes the proof.
\end{proof}

\subsubsection{Proof of Lemma \ref{lemma-small-pay-add}: Observations About The Mechanism} \label{subsubsec-prop-alloc-add}
The proof of \cref{lemma-small-pay-ud} is a direct consequence of the approximation guarantee of the mechanism and the fact that it is 
obviously strategy-proof (and thus also dominant-strategy incentive compatible) and satisfies individual rationality and no negative transfers. 

\begin{proof}[Proof of \cref{lemma-small-pay-add}]
The proofs of parts \ref{item-1-add}, \ref{item-2-add}, \ref{item-4-add} and \ref{item-4.5-add} are identical to the respective parts in the proof of \cref{lemma-small-pay-ud}.

For part \ref{item-6-add}, we first show that bidder $1$ wins a bundle that contains item $a$ given $(v_1^{both},v_2^{b,one})$. Note that since the allocation rule $f$ and the payment scheme $P_1$ are realized by a dominant-strategy mechanism, we have that:
\begin{equation}\label{eq-part6-add}
\begin{aligned}
    v_1^{both}(f(v_1^{both},v_2^{b,one}))-P_1(v_1^{both},v_2^{b,one})&\ge v_1^{both}(f(v_1^{a,one},v_2^{b,one}))-P_1(v_1^{a,one},v_2^{b,one}) \\
    &\ge 2k+2 &\text{(by part \ref{item-1-add})}
\end{aligned}    
\end{equation}
Combining (\ref{eq-part6-add}) with the property of no negative transfers implies that $v_1^{both}(f(v_1^{both},v_2^{b,one}))\ge 2k+2$. Thus, given $(v_1^{both},v_2^{b,one})$, bidder $1$ necessarily gets a bundle that contains item $a$. 
For the upper bound on the payment, note that $v_1^{both}(f(v_1^{both},v_2^{b,one}))\le 4k+4$, so by individual rationality  $P_1(v_1^{a,both},v_2^{b,one})\le 4k+4$. By that, we  complete the proof of part \ref{item-6-add}. 

For part \ref{item-complicated-add},
we begin by showing that given $(v_1^{both},v_2^{a,large})$, bidder $1$ does not win item $a$. This is due to weak monotonicity. Formally, we remind that by part \ref{item-2-add}, given $(v_1^{a,mid},v_2^{a,large})$, bidder $1$ wins a bundle $S^{mid}$
that does not contain $a$. Thus, if $f(v_1^{both},v_2^{a,large})$ allocates 
a bundle $S^{both}$ that contains item $a$ to bidder $1$, then we have that $v_1^{a,mid}(S^{both})-v_1^{a,mid}(S^{mid})=k^2$ whereas $v_1^{both}(S^{both})-v_1^{both}(S^{mid})\le 4k+4$, so $f$ is not weakly monotone. Since $f$ and $P_1$ realize a dominant-strategy mechanism, \cref{wmon-lemma} implies that $f$ has to be weakly-monotone, so we get a contradiction. Thus, $f(v_1^{both},v_2^{a,large})$ has to allocate bidder $1$ a bundle that does not contain item $a$. 

The bounds on the payments are a straightforward implication of individual rationality and no negative transfers. If $f(v_1^{both},v_2^{a,large})$ allocates item $b$ to bidder $1$ then the payment is at most $2k+1$, and if it allocates to bidder $1$ no valuable items, then the payment has to be at most zero. Due to no negative transfers, it is zero exactly.   
\end{proof}

\subsection{Proofs of Claims \ref{cl:subadditive} and \ref{cl:general}} \label{subsec::proofs-subadd-general}
\begin{proof}[Proof of \cref{cl:subadditive}]
First the bidders are assigned a an arbitrary order. 
This is clearly OSP.  Then, after randomly partitioning the bidder into two groups $G_1$ and $G_2$ (which is done uniformly at random and independently of bidder valuations) with probability $1/2$ the auction runs a second-price auction for the grand bundle on bidders in $G_1$.  This is clearly OSP for these bidders.  

With the remaining probability bidders in $G_1$ are discarded and the auction learns the highest value among these bidders for the grand bundle (this is also clearly OSP since bidders in $G_1$ do not gain any positive utility under any value profile).  After learning this value in an OSP way from bidders in $G_1$ the auction runs the \textsc{Binary-Search-Mechanism} on the bidders in $G_2$.  Recall that this mechanism draws independently and uniformly at random a round $r_i$ for each $i \in G_2$ and a random final round $r^*$ (before any bidder in $G_2$ acts).  

Then in each round, the mechanism sets a price for each item, allows the bidders in that round (in the pre-specified order generated at the outset of the auction) to state their demand set for ``unclaimed'' items and ``conditionally claim'' them, and then if the current round is $r^*$, terminate the auction awarding bidders in $r^*$ their conditionally claimed items (otherwise, no bidders in the round are allocated any items).  Observe that each bidder in $G_2$ participates in exactly one round and is, thus, asked a single demand query.  If the round $r_i$ that bidder $i$ participates in is $r^*$ then she weakly maximizes her utility when she is called to act by truthfully reporting her demand set (since she is allocated exactly these items).  On the other hand, if $r_i \neq r^*$ then she obtains no utility under any possible valuation profile of all agents and, thus, truthfully reporting her demand set is a weakly obviously dominant strategy.  

Since for any fixed realization of the random outcomes of all coin flips the resulting mechanism is OSP, we have that the randomized mechanism is universally OSP.
\end{proof}
\begin{proof}[Proof of \cref{cl:general}]
    For this auction, bidders are randomly partitioned into three groups \texttt{STAT}, \texttt{SECOND-PRICE}, and \texttt{FIXED}.  We argue that for any fixed partition the mechanism is OSP and, hence, the mechanism is universally OSP. 
    
    Bidders in \texttt{STAT} cannot possibly win any goods and the mechanism just requests that they report their values.  Since they win nothing regardless of report, it is an obviously dominant strategy for them to report their information truthfully. 
    Bidders in \texttt{SECOND-PRICE} participate in a second-price auction for the grand bundle.  Since this can be implemented as an ascending-price auction for the grand bundle,  \cref{lemma-partial} 
    implies that bidders in \texttt{SECOND-PRICE} have an obviously dominant truthful strategy.  
    Finally, bidders in \texttt{FIXED} are approached in an arbitrary order and they are allocated their utility maximizing bundle of the remaining items given a fixed vector of prices.  
    Thus, 
    each
     bidder in \texttt{FIXED} an obviously dominant truthful strategy.  
    
    As such, each bidder on any fixed realization of the random partition has an obviously dominant truthful strategy and the entire mechanism is then universally OSP.
\end{proof}

\end{document}